\documentclass[conference]{IEEEtran}
\IEEEoverridecommandlockouts
\usepackage{cite}
\usepackage{amsmath,amssymb,amsfonts}
\usepackage{algorithmic}
\usepackage{graphicx}
\usepackage{textcomp}
\usepackage{xcolor}
\def\BibTeX{{\rm B\kern-.05em{\sc i\kern-.025em b}\kern-.08em
    T\kern-.1667em\lower.7ex\hbox{E}\kern-.125emX}}

\usepackage{hyperref}

\usepackage{microtype}
\usepackage{subfig}
\usepackage{booktabs}

\usepackage{amsmath,amssymb,amsthm}
\usepackage{subfig}
\usepackage[format=plain,labelsep=period, font = small]{caption}
\usepackage{graphicx}

\newcommand{\spara}[1]{\smallskip\noindent{\bf{#1}}}

\newcommand{\squishlist}{\begin{list}{$\bullet$}
  { \setlength{\itemsep}{0pt}
     \setlength{\parsep}{3pt}
     \setlength{\topsep}{3pt}
     \setlength{\partopsep}{0pt}
     \setlength{\leftmargin}{1.5em}
     \setlength{\labelwidth}{1em}
     \setlength{\labelsep}{0.5em} } }
\newcommand{\squishend}{
\end{list}  }

\newcommand{\ptitle}[1]{\vspace{1mm}\noindent{\bf #1.}}
\newcommand{\norm}[1]{\big\lVert#1\big\rVert}

\DeclareMathOperator*{\argmax}{arg\,max}

\newtheorem{theorem}{Theorem}
\newtheorem{lemma}{Lemma}
\newtheorem{proposition}{Proposition}

\begin{document}

\title{Block-Approximated Exponential Random Graphs\\
\thanks{\textsuperscript{*}This work was done while the author was with Ghent University and presented at the DSAA 2020 conference.}
}

\author{\IEEEauthorblockN{Florian Adriaens\textsuperscript{*}}
\IEEEauthorblockA{\textit{KTH Royal Institute of Technology, Sweden} \\
first.last@gmail.com}
\and
\IEEEauthorblockN{Alexandru Mara\quad Jefrey Lijffijt\quad Tijl De Bie}
\IEEEauthorblockA{\textit{IDLab, Ghent University, Belgium} \\
first.last@ugent.be}
}

\maketitle

\begin{abstract}
An important challenge in the field of exponential random graphs (ERGs) is the fitting of non-trivial ERGs on large graphs.
By utilizing fast matrix block-approximation techniques, we propose an approximative framework to such non-trivial ERGs that result in dyadic independence (i.e., edge independent) distributions, while being able to meaningfully model \emph{local} information of the graph (e.g., degrees) as well as \emph{global} information (e.g., clustering coefficient, assortativity, etc.) if desired. 
This allows one to efficiently generate random networks with similar properties as an observed network, and the models can be used for several downstream tasks such as link prediction.
Our methods are scalable to sparse graphs consisting of millions of nodes.

Empirical evaluation demonstrates competitiveness in terms of both speed and accuracy with state-of-the-art methods---which are typically based on embedding the graph into some low-dimensional space--- for link prediction, showcasing the potential of a more direct and interpretable probablistic model for this task.
\end{abstract}


\section{Introduction}
Network modeling is typically concerned with the following setting: given measurements on certain structural properties of a real-life network, such as node degrees, clustering coefficient, density, and so forth, one wishes to find a model for the network where samples generated by the model have similar values of the measured properties.
Prominent examples include stochastic block models \cite{NIPS2008_3578}, graphons \cite{borgs2014an} and exponential random graphs (ERGs) \cite{Holland}.

ERGs have received significant attention in numerous research areas \cite{Holland,Strauss,Snijders,Robins}, as they provide a well-founded probabilistic framework to network modeling; their study has been described as the ``statistical mechanics of networks'' \cite{ParkNewman}.
They arise naturally as the solution to the problem of finding the maximum entropy distribution over a set of graphs, while respecting constraints that a probabilistic statistic must equal an observed statistic in expectation.
Equivalently, the parameters of an ERG can be determined by maximum likelihood estimation of an exponential distribution for explaining the observed statistics in a principled manner.
The advantage of ERGs is that they can represent a wide range of structural tendencies, such as transitivity and degree heterogeneity, by capturing complicated dependence patterns that are not easily modeled by simpler probabilistic models.

The downside of classically specified ERGs is their limited scalability \cite{Goodreau}. Inferring the parameters of an ERG is often intractable for networks of even moderate size, because computing the partition function may require evaluating a summation over all $2^{\Theta(n^2)}$ for graphs with $n$ nodes.
The common approach is to approximately infer the parameters by MCMC sampling, after which a goodness of fit is measured by generating random networks from the ERG and comparing statistics with the observed network. Besides limited scalability, MCMC methods often have the problem of degeneracy \cite{Handcock}, assigning most probability mass to either near-empty or near-full graphs, for example when one is counting Markov neighborhood properties such as triangles \cite{Snijders}. 

ERGs typically lose efficiency when they aim to model global constraints that suggest dependencies between edges, such as degree assortativity or the existence of many triangles. Modeling such information is important in practice: e.g. if the number of triangles is large as compared to other graphs with the same local properties, this may indicate that edge formation is a result of a triadic closure process, which can be an important property to understand and model.

Addressing this challenge, in this paper we propose the use of fast matrix block-approximation techniques to derive intelligible and scalable approximations to non-trivial ERGs for large sparse graphs ($>10^6$ nodes). The resulting models are dyadically independent, while still meaningfully incorporating local information (degrees) and global information (triangles, assortativity, etc.).

\spara{Contributions and roadmap.}
\squishlist
\item Discussion of a general dyadic independence model that is able to incorporate arbitrary (structural) features between pairs of nodes (Section~\ref{sec:entropymodels}).
\item Depending on the incorporated features, fitting the model parameters may not be scalable. We propose fast matrix block-approximation techniques to provide scalable approximations to the general model (Section~\ref{sec:lowrankapp}).
\item An empirical comparison with several state-of-the-art (embedding-based) methods for the task of predicting missing links, and a goodness-of-fit comparison with two well-known probablistic models (Section~\ref{sec:exp}).
\squishend

\section{MaxEnt models with structural constraints}\label{sec:entropymodels}
\ptitle{Notation}
Let $G=(V,E)$ be a graph with $|V|=n$ nodes and $|E|=m$ edges. The neighborhood of node $i$ is denoted as $\mathcal{N}(i)$.
The adjacency matrix of an observed graph $G$ is denoted as $\hat{\mathbf{A}} =[\mathbf{\hat{A}}_{ij}] \in \mathcal{A}$, where $\mathcal{A} = \{0,1\}^{n\times n}$ is the set of square matrices of size $n$.
A random matrix is denoted as $\mathbf{A} \in \mathcal{A}$.
The transpose of a matrix $\mathbf{A}$ is denoted as $\mathbf{A}^T$ and the Frobenius norm is denoted as $\norm{\mathbf{A}}$.
The expected value operator related to a distribution $P$ is denoted as $\mathbb{E}_P[\cdot]$.
The focus in this paper is restricted to undirected graphs, but extensions to directed networks are straightforward.

\ptitle{ERGs as maximum entropy (MaxEnt) models}
The setting of this paper is that, based on some measurements (statistics) of an observed graph $\hat{\mathbf{A}}$, we wish to find a probability distribution over $\mathcal{A}$, i.e. the set of all possible graphs over the same set of vertices.
ERGs are a family of statistical models that are often used for this task.
Let $s(\cdot)$ be the vector of statistics of a graph, e.g. one component of $s(\cdot)$ could be the count of all triangles.
An ERG with associated parameter vector $\theta$ is then typically introduced by positing an exponential distribution over $\mathcal{A}$:
\begin{align}
\label{eq:ergintro}
P_{\theta}(\mathbf{A}=A)= \exp(\theta^Ts(A)-\psi),
\end{align}
where $\psi$ denotes a normalizing constant.
Given an observed graph $\hat{\mathbf{A}}$, the parameters $\theta$ are then typically derived by maximum likelihood estimation (MLE).

An equivalent, but less-known, way of introducing ERGs is by looking for the maximum entropy distribution over $\mathcal{A}$, subject to constraints that the expected statistics are equal to the observed statistics. The MLE of (\ref{eq:ergintro}) is equal to this maximum entropy distribution \cite{hillar2013maximum}.
Throughout the paper, we will take the maximum entropy point-of-view.


\subsection{An example with degree assortativity\label{sec:exdegreeassort}}
We start by formalizing what we \emph{wish} to solve, and end up relaxing these equations in order to have analytical expressions containing parameters that are tractable to infer.
As an example, suppose we are interested in finding the maximum entropy (MaxEnt) distribution over $\mathcal{A}$, subject to a constraint in expectation on each individual node degree, as well as a constraint in expectation on the total \emph{degree assortativity} in the graph, as measured by the sum of $|\mathcal{N}(i)|\cdot|\mathcal{N}(j)|$ over all edges $(i,j) \in E$. We aim to solve\footnote{We implicitly assume $\textstyle\sum P(\textbf{A}) = 1$ and $P(\textbf{A})\geq 0$ in all the MaxEnt problem formulations, but omit them for brevity.}
\begin{align}
\label{eq:maxent}
\displaystyle &\argmax_{P(\textbf{A})} \quad -\displaystyle \mathbb{E}_P[\log P(\textbf{A})], \\
\text{s.t. }\quad &\mathbb{E}_P[\textstyle\sum_{i,j} \textbf{A}_{ij}\textstyle\sum_c \textbf{A}_{ic}\textstyle\sum_r \textbf{A}_{rj}] = c, \notag\\
&\mathbb{E}_P[\textstyle \sum_{j}\textbf{A}_{ij}] = d_i \quad \forall i,\notag\\
&\mathbb{E}_P[\textstyle\sum_{i}\textbf{A}_{ij}] = d_j \quad \forall j,\notag
\end{align}
where $d_i=\sum_{j}\mathbf{\hat{A}}_{ij}$ is the observed degree of node $i$ in $G$, and $c=\sum_{i,j} \mathbf{\hat{A}}_{ij}\sum_c \mathbf{\hat{A}}_{ic}\sum_r \mathbf{\hat{A}}_{rj}$ is the observed assortativity measure.
A typical difficulty with this classically specified ERG is that in general the normalizing constant of the distribution $P$ is infeasible to compute, because of edge dependencies introduced by the assortativity constraint. As such, the parameters of the distribution are intractable to compute exactly.
This is in contrast to a MaxEnt model subject to only degree constraints. As observed by different authors \cite{ParkNewman, DeBie2011, Parisi}, in this case $P$ is a dyadic independence model: it factorizes as a product of independent Bernoulli distributions, one for each node pair $(i,j) \in V \times V$.
Moreover, in a degree-only model, the number of unique parameters to be optimized over is fully determined by the number of \emph{unique} degrees in $G$.
As shown by \cite{DeBie2011}, for sparse $G$ (where $m=O(n)$), the problem has in fact only $O(\sqrt{n})$ free variables instead of the $O(n)$ original variables (one original variable for each degree constraint), making inference possible on very large networks.
In Section~\ref{sec:lowrankapp}, we show that we can do a similar reduction in variables also for more complex models.

\begin{figure}[tp]
\captionsetup[subfloat]{farskip=5pt,captionskip=1pt}
\captionsetup[subfigure]{justification=centering}
\centering
\subfloat[\label{fig:kar1}]{\scalebox{0.66}{\includegraphics[width=0.7\linewidth,trim=0cm 0cm 0cm 0cm,clip]{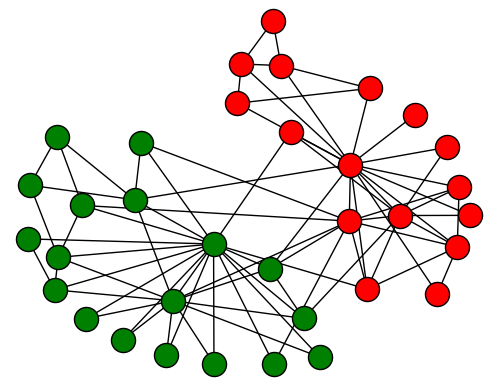}}}
\\
\subfloat[\label{fig:kar2}]{\scalebox{0.92}{\includegraphics[width=0.52\linewidth,trim=1cm 1cm 1cm 0.5cm,clip]{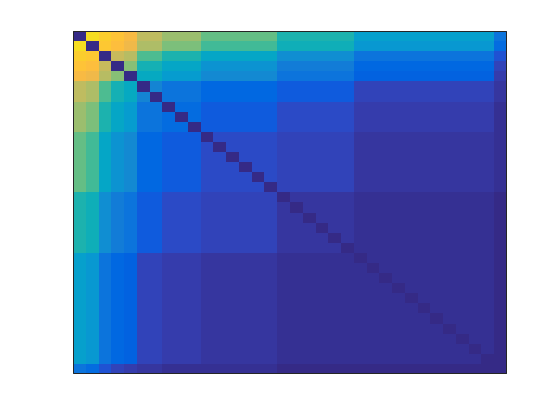}}}
\subfloat[\label{fig:kar3}]{\scalebox{0.92}{\includegraphics[width=0.52\linewidth,trim=1cm 1cm 1cm 0.5cm,clip]{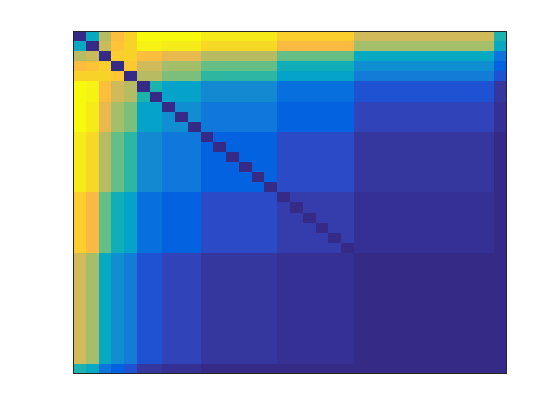}}}
\\
\subfloat{\includegraphics[width=0.6\linewidth,trim=0cm 0.1cm 0cm 0.15cm,clip]{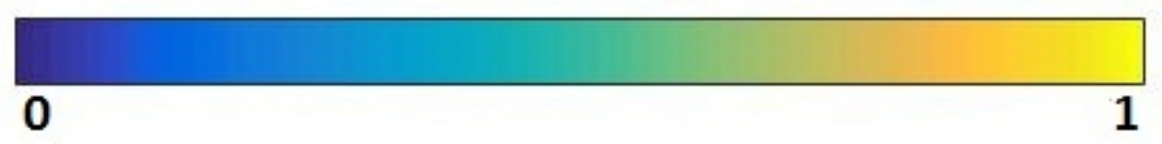}}
\caption{(a) The \emph{Karate} network (b) MaxEnt edge succes probabilities with constraints on degrees (c) MaxEnt edge probabilities with constraints on degrees, and a global constraint on degree assortativity as given by (\ref{eq:maxent2}).
In both heatmaps, nodes are sorted in descending degree order (highest degree in the top left corner).\label{fig:Karate2}}
\end{figure}

One way to approximate (\ref{eq:maxent}) is to replace the row- and column sums of random variables $\sum_c\textbf{A}_{ic}$ by their expectation:
\begin{equation*}
\sum_c\textbf{A}_{ic} \approx \mathbb{E}_P[\sum_c\textbf{A}_{ic}] = d_i.
\end{equation*}
The approximated model of (\ref{eq:maxent}) then becomes
\begin{align}\label{eq:maxent2}
\displaystyle &\argmax_{P(\textbf{A})} \quad -\displaystyle E_P[\log P(\textbf{A})], \\
\text{s.t. }\quad &\mathbb{E}_P[\textstyle\sum_{i,j} d_id_j\textbf{A}_{ij}] = c, \notag\\
& \mathbb{E}_P[\textstyle\sum_{j}\textbf{A}_{ij}] = d_i \quad \forall i,\notag\\
& \mathbb{E}_P[\textstyle\sum_{i}\textbf{A}_{ij}] = d_j \quad \forall j.\notag
\end{align}
We omit the details, but it is easy to show that the solution to (\ref{eq:maxent2}) is again a dyadic independence model, and the complexity depends on the number of unique degrees in $G$.

Let us provide some intuition on the solution of (\ref{eq:maxent2}).
First, because the distribution has maximum entropy, it is the unique distribution that injects no side information on properties that were not taken into account as constraints \cite{Cover}.
Secondly, if the degree assortativity cannot be explained by a model that is inferred using \emph{only} degree constraints, for example if $c$ is larger than expected under a model where only degrees are constrained, then the optimum of (\ref{eq:maxent2}) will on average assign higher probabilities between pairs of high degree nodes and between pairs of low degree nodes, at the expense of pairs of nodes where one has a large and the other a low degree.
This results in a more accurate fit of the observed graph.

An example of a highly disassortative network is the \emph{Karate} dataset \cite{Karate}. It has a negative assortativity coefficient \cite{NewmanAssort} of $-0.48$, and thus nodes with similar degrees are less often connected.
This is confirmed in a visualization of the network in Figure~\ref{fig:kar1}.
The dataset essentially consists of two Karate club teachers (the two highest-degree nodes), mostly connected to their own students, with few edges between the two communities.
Most connections are between a high-degree node and a low-degree node, and the two teachers themselves are not connected.

Figure~\ref{fig:kar2} shows the edge probabilities between all pairs of nodes, when the network is modeled by a MaxEnt model with a constraint on the expectation of each individual node degree (the degree-only model).
Connections between high degree nodes are more likely, and thus it assigns most probability mass in the top left corner.
On the other hand, Figure~\ref{fig:kar3} shows the MaxEnt model with an additional assortativity constraint as given by (\ref{eq:maxent2}). It builds on the degree-only model by taking into account the network's disassortativity, and as such it lowers the edge probabilities between nodes of similar degree and increases the edge probabilities between nodes with dissimilar degrees.
Note that it correctly assigns a low probability for a connection between the two Karate teachers.

\subsection{Generalizing to arbitrary features}\label{sec:genmodel}
Building on the previous example, we can view (\ref{eq:maxent2}) as a model that takes into account \emph{observed} features $f_{ij} \triangleq d_id_j$.
Instead of taking the product of two node degrees, we can simply extend this to arbitrary observed features.\footnote{Such a generalization is often less well-founded than the case of degree assortativity (where we replaced a sum of random variables with the sum of their expectations), but these models still turn out to be useful.}
Denote $\mathbf{F} = [f_{ij}] \in \mathbb{R}^{n \times n}$ as an induced pairwise feature matrix.
For example,  the degree matrix of a node $i$ is a matrix with ones on the $i$-th row and zeros elsewhere.
The common neighbor matrix (triangle counting) is given by  $\mathbf{F}_\text{CN} \triangleq \mathbf{\hat{A}}^2 = [|\mathcal{N}(i) \cap \mathcal{N}(j)|]$.
Counts of different types of graphlets besides triangles can be incorporated as well, e.g., by using so-called `weighted motif graphs'  \cite{AhmedEsti, Hone} as feature matrices.

\begin{figure}[tp]
\captionsetup[subfloat]{farskip=5pt,captionskip=1pt}
\captionsetup[subfigure]{justification=centering}
\centering
\subfloat[\label{fig:dola}]{\scalebox{1}{\includegraphics[width=0.5\linewidth,trim=0.2cm 1cm 0.2cm 1cm,clip]{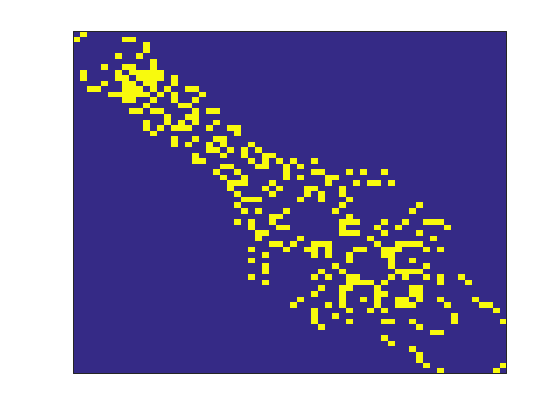}}}
\subfloat[\label{fig:dolb}]{\scalebox{1}{\includegraphics[width=0.5\linewidth,trim=0.2cm 1cm 0.2cm 1cm,clip]{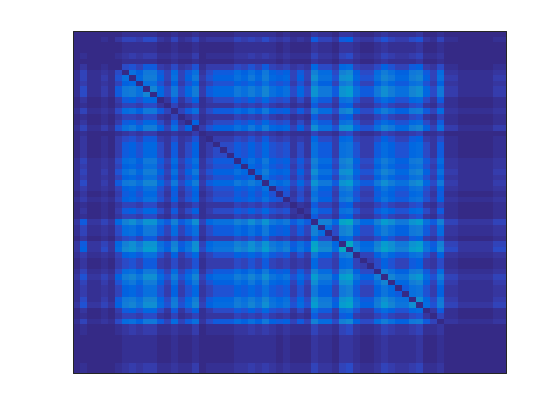}}}
\\[-1ex]
\subfloat[\label{fig:dolc}]{\scalebox{1}{\includegraphics[width=0.5\linewidth,trim=0.2cm 1cm 0.2cm 0.5cm,clip]{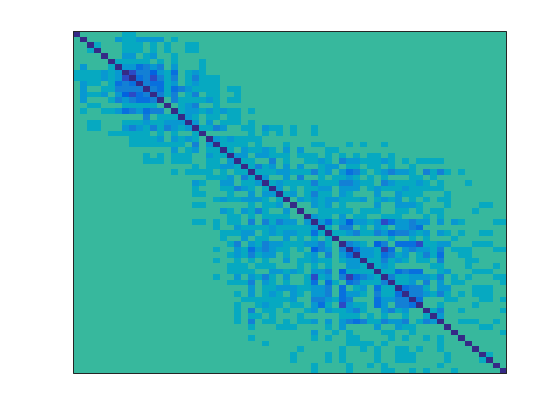}}}
\subfloat[\label{fig:dold}]{\scalebox{1}{\includegraphics[width=0.5\linewidth,trim=0.2cm 1cm 0.2cm 0.5cm,clip]{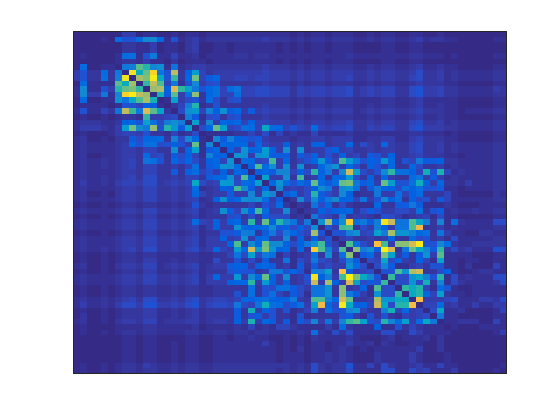}}}\\
\subfloat{\includegraphics[width=0.6\linewidth]{colorbar.pdf}}
\caption{MaxEnt edge succes probabilities on the \emph{Dolphins} network (a), when different kinds of constraints are taken into account; (b) MaxEnt with degrees (c) MaxEnt with only a global constraint on $\mathbf{F}_\text{CN}$ (d) MaxEnt with a combined constraint on both the degrees and $\mathbf{F}_\text{CN}$, demonstrating that combining local and global information results in a superior and more realistic model. \label{fig:dolphins}}
\end{figure}

Given $M$ of these feature matrices $\mathbf{F}_l \triangleq [f^l_{ij}]$ with $l=1,\ldots,M$, we aim to solve
\begin{align}\label{eq:maxentGeneral}
\displaystyle &\argmax_{P(\textbf{A})} \quad -\displaystyle \mathbb{E}_P[\log P(\textbf{A})], \\
\text{s.t. }\quad &\mathbb{E}_{P}[\textstyle\sum_{i,j} f^l_{ij}\textbf{A}_{ij}] =  c_l  \quad \forall l=1,\ldots,M.\notag
\end{align}
Where $c_l = \textstyle\sum_{i,j} f^l_{ij}\mathbf{\hat{A}}_{ij} = \sum_{(i,j) \in E} f^l_{ij}$ are the observed statistics.
Solutions to the class of models defined by (\ref{eq:maxentGeneral}) have the convenient property that the partition function (the normalizing constant) factorizes as a product over all possible edges.
As such (\ref{eq:maxentGeneral}) is again a dyadic independence model, with a Bernoulli success probability for a node pair $(i,j)$:
\begin{align*}
P(\textbf{A}_{ij}=1) = \dfrac{\exp(\textstyle\sum_{l=1}^M f_{ij}^l \lambda_l)}{1+\exp(\textstyle\sum_{l=1}^M f_{ij}^l \lambda_l)},
\end{align*}
where $\lambda_l$ denotes the Lagrange multiplier associated with the $l$-th constraint in (\ref{eq:maxentGeneral}).  These multipliers are found by unconstrained minimization of the convex Lagrange dual function:
\begin{align}\label{eq:lagrangian}
&L(\lambda_1,\ldots,\lambda_M) =\nonumber\\
&\textstyle\sum_{i,j}\log(1+\exp(\textstyle\sum_{l=1}^M f_{ij}^l\lambda_l))-\sum_{l=1}^M c_l \lambda_l.
\end{align}
The partial derivatives $\forall l=1,\ldots,M$ are computed as
\begin{align*}
\dfrac{\partial L}{\partial \lambda_l} = \textstyle\sum_{i,j} f_{ij}^l\cdot\dfrac{\exp\big(\textstyle\sum_{l=1}^M f_{ij}^l\lambda_l\big)}{1+\exp\big(\textstyle\sum_{l=1}^M f_{ij}^l\lambda_l\big)}-c_l.
\end{align*}

Instead of counting exact specifications in the original ERG, the models (\ref{eq:maxentGeneral}) can be seen as applying a `linear mask' in order to approximately count the specifications.
In fact, they are often identical to expressions found by pseudo likelihood estimation \cite{Strauss2,DuijnERG}.
Pseudo likelihood is mainly used for estimation of the original ERG parameters \cite{DuijnERG}, but we argue that dyadic independence models can be valuable by themselves.
For the exact derivations and a discussion regarding the pseudo likelihood, we refer to Appendix A.

In this paper, we focus on combining individual local constraints (degrees) with a limited number of global structural constraints.\footnote{Two or three global constraints often suffice to obtain qualitative models.}
Figure~\ref{fig:dolphins} shows an example of the predictive power of combining constraints on degrees with a global structural constraint $\mathbf{F}_\text{CN}$ (triangle count) on the \emph{Dolphins} dataset \cite{dolphins}.
Figure~\ref{fig:dola} shows the connectivity of the dataset.
Figure~\ref{fig:dolb} shows the edge probabilities according to a degree-only MaxEnt model. It assigns higher probability to edges between two high degree nodes, but it fails to capture any form of local community structure.
Figure~\ref{fig:dolc} is a MaxEnt model fitted with just the structural constraint $\mathbf{F}_\text{CN}$. Although it captures community structure, it fails to make good predictions about the actual edges in the network. Indeed, most observed edges have a low probability of being present,
and all node pairs with zero common neighbors are assigned a probability of 1/2 of being connected.
Figure~\ref{fig:dold} shows the model when combining degree constraints and the structural constraint $\mathbf{F}_\text{CN}$.
It leads to a remarkably good fit of the original network, while still leaving room for prediction and inference.

\subsection{Scalability issues}\label{sec:scalabilityissues}
Inference on large graphs is typically not possible when we have $M=O(n)$ constraints, as is the case when combining degree constraints with a limited number of global constraints.
Minimizing (\ref{eq:lagrangian}) can be viewed as a classical learning setting with $n^2$ training examples and $O(n)$ weights.
Standard gradient methods need $\Omega(n^2)$ computations per iteration and $\mathbf{F}$ needs to be stored in memory.
To resolve both issues, in Section~\ref{sec:lowrankapp} we propose to \emph{block-approximate} the feature matrices. In particular, Theorem~\ref{lem:redvar} shows that in sparse graphs $(m=O(n))$, the number of variables to be optimized over roughly reduces from $O(n)$ to $O(\sqrt{n})$. At the same time, block-approximated matrices are easily maintained in memory.


\subsection{Label leakage \& degeneracy}\label{sec:featureleakage} 
To avoid overfitting, it is crucial to avoid `label leakage'.
More concretely, when predicting whether or not an edge exists, one should avoid making direct use of the actual existence of that edge while fitting the model.
Indeed, if one selects $\mathbf{F}=\mathbf{\hat{A}}$, the only solution to (\ref{eq:maxentGeneral}) is exactly $P(\mathbf{\hat{A}})=1$.
In other words, the model of the network is the network itself, rendering it completely useless for tasks like link prediction.
In contrast, $\mathbf{F}_\text{CN}$ is an excellent candidate for a link prediction feature matrix, since the number of common neighbors between two nodes is not directly related to the existence of that edge, and if it is the case, the model will actually learn the relation.
Remarkably, other methods often overlook this fact. For example, the method by \cite{Zhang2018ArbitraryOrderPP} considers embeddings defined by a truncated singular value decomposition of the adjacency matrix $\mathbf{\hat{A}}$.
The Katz centrality measure \cite{Katz1953} is another such example.
Nevertheless, for other tasks that require a closer fit to the observed network (e.g., network reconstruction), it might be useful to include adjacency information, since it allows to become arbitrarily close to the observed network.
For example, $\mathbf{F} = \mathbf{\hat{A}}+\beta\mathbf{\hat{A}}^2$ with small $\beta>0$ was used in Section~\ref{sec:gof} to obtain a close fit to the \emph{Facebook} network.

The other side of degeneracy, i.e. assigning most probability mass to near-empty graphs, is seemingly avoided by incorporating degree constraints into the model. For example, compare Figure~\ref{fig:dolc} with Figure~\ref{fig:dold}.

\section{Block-approximating feature matrices}\label{sec:lowrankapp}

\subsection{Motivation}\label{sec:motivation}
Prior work \cite{Adriaens2019, ecmltrees} on improving the scalability of fitting a general MaxEnt model (\ref{eq:maxentGeneral}) looks for permutations of $(\lambda_1,\ldots,\lambda_M)$ that leave the Lagrange dual function $L$ (\ref{eq:lagrangian}) invariant.
The convexity of $L$ then implies that if there is a permutation that maps $\lambda_i$ to $\lambda_j$, then there necessarily exists an optimum of (\ref{eq:lagrangian}) where $\lambda_i=\lambda_j$ \cite[Section 4.2]{Adriaens2019}.
Similarly as in \cite{Adriaens2019}, we look for equivalent Lagrange multipliers associated with the degree constraints.
Equivalent Lagrange multipliers are equated and the \emph{reduced model} is solved by standard convex optimization methods.

However, for general feature matrices equivalences are rare.
Thus, in this paper we propose to block-approximate the feature matrices. Let $\overline{\mathbf{F}} \in \mathbb{R}^{n\times n}$ be a block-approximation of $\mathbf{F}$, represented by a structure with $k \times k$ blocks, with each block being a submatrix with constant values.
Theorem~\ref{lem:redvar} states that, when fitting a MaxEnt model (\ref{eq:maxentGeneral}) with constraints on each node degree, as well as a global constraint as given by $\overline{\mathbf{F}}$, the number of free variables in the reduced model is at most $\sqrt{2km}+1$ instead of the original $n+1$ variables ($n$ degrees and one global constraint).
This implies a significant speed-up when fitting such MaxEnt models on large sparse graphs.


\begin{lemma}\label{lem:sparse}
Let $\mathbf{\hat{A}}$ be symmetric with $m$ non-zero entries. If the nodes are partioned into $k$ disjoint groups, then the sum over all groups of the number of unique degrees inside each group is at most $\sqrt{2km}$.
\end{lemma}

\begin{proof}
The nodes are partitioned into $k$ disjoint groups.
Denote $m_i$ as the number of ones in the rows of $\mathbf{\hat{A}}$, when the rows are restricted to the nodes in group $i$ ($1\leq i \leq k$).
Let $u_i$ denote the number of unique degrees in group $i$.
Then it holds that \cite[Lemma 2]{DeBie2011}:
\begin{equation*}
u_i \leq \sqrt{2\cdot m_i}
\end{equation*}

Summating over the $k$ groups:
\begin{align*}
\sum_{i=1}^k u_i \leq \sqrt{2}\cdot(\sqrt{m_1}+\ldots+\sqrt{m_k}).
\end{align*}
Utilizing $m_1+\ldots+m_k=m$, the concavity of $\sqrt{\cdot}$ and Jensen's inequality:
\begin{align*}
\sqrt{m_1}+\ldots+\sqrt{m_k} \leq k\cdot \sqrt{\frac{m}{k}}.
\end{align*}
Hence $\sum_{i=1}^k u_i \leq \sqrt{2km}$ for sparse matrices.
\end{proof}

\begin{theorem}\label{lem:redvar}
Let $\mathbf{\hat{A}}$ be symmetric with $m$ non-zero entries. Let $\overline{\mathbf{F}}$ be a $k \times k$ blockmatrix, with each block being a submatrix with constant entries. A MaxEnt model (\ref{eq:maxentGeneral}), with constraints on each individual degree and a constraint as given by $\overline{\mathbf{F}}$, can be solved by optimizing an unconstrained convex problem with at most $\sqrt{2km}+1$ variables.
\end{theorem}

\begin{proof}
Let $\lambda_i$ be the Lagrange multiplier associated with the degree constraint of node $i$. Observe that the function $L$ $(\ref{eq:lagrangian})$ is invariant if one swaps $\lambda_i$ with 
$\lambda_j$, when $i$ and $j$ are part of the same block \emph{and} have the same degree $d_i=d_j$.
Convexity of $L$ then implies there exists an optimum where $\lambda_i = \lambda_j$ \cite[Section 4.2]{Adriaens2019}. As such, the number of free variables is equal to the number of unique degrees in each block (plus one variable for the global constraint).
By Lemma~\ref{lem:sparse}, the total number of unique degrees summated over $k$ blocks is bounded by $\sqrt{2km}$.
\end{proof}

\subsection{Methods}\label{sec:methodslowrank}
In this section, we discuss several methods for obtaining fast and qualitative block-approximations for certain classes of feature matrices.
Most methods rely on a fast top d eigendecomposition of the---assumed to be sparse---adjacency matrix $\hat{\mathbf{A}}$.

\subsubsection{Spectral clustering adjacency polynomials}\label{sec:hop}
A principled way of block-approximating a feature matrix $\mathbf{F}$ is to first spectral cluster $\mathbf{F}$ and then replacing rows with centroids.
Assuming $\mathbf{F}$ is real and symmetric, spectral clustering first calculates a truncated eigenvalue decomposition to get an optimal low-rank approximation of $\mathbf{F}$, after which k-means is used to cluster the nodes in the low dimensional space \cite{Ulrike}.
A block-approximation of $\mathbf{F}$ is then found by
\begin{align}\label{fsimap}
\mathbf{F} \approx \mathbf{U}\mathbf{S}\mathbf{U}^T \approx \mathbf{U}_c\mathbf{S}\mathbf{U}_c^T,
\end{align}
with $\mathbf{U}, \mathbf{U}_c \in \mathbb{R}^{n \times d}$ and $\mathbf{S} \in \mathbb{R}^{d \times d}$ a diagonal matrix containing the signs of the $d$ largest (in absolute value) eigenvalues of $\mathbf{F}$. The matrix $\mathbf{U} = \mathbf{V}\sqrt{\mathbf{\Sigma}}$ consists of the rescaled eigenvectors $\mathbf{V}$ corresponding to the top $d$ eigenvalues, where $\mathbf{\Sigma}$ denotes a diagonal matrix with the absolute values of these eigenvalues.  The matrix $\mathbf{U}_c$ is defined by replacing each row of $\mathbf{U}$ by its respective cluster centroid.
Notice that symmetry of $\mathbf{F}$ is maintained by both of the approximations in (\ref{fsimap}).
Now let $\overline{\mathbf{F}}\triangleq \mathbf{U}_c\mathbf{S}\mathbf{U}_c^T$.
The following proposition gives a bound on the expected distance between $\overline{\mathbf{F}}$ and $\mathbf{F}$.\footnote{This should be a standard result in spectral clustering theory, but the proof is given for completeness.}

\begin{proposition}\label{thm:bound}
Let $|\lambda_1| \geq |\lambda_2| \geq \ldots \geq |\lambda_d|$ be the $d$ largest eigenvalues of $\mathbf{ F}$.
Let $\phi_{OPT,k}$ denote the optimal clustering objective value with $k$ clusters.
Using k-means++ as a (randomized) clustering algorithm, the expected error $\mathbb{E}\big[\norm{\mathbf{F}-\overline{\mathbf{F}}}\big]$ is at most
\begin{equation*}
2\sqrt{\textstyle \sum_{i=1}^{d}|\lambda_i|}\cdot O(\emph{log}(k))\cdot \phi_{OPT,k}
\end{equation*}
additively larger than any optimal rank $d$ approximation of $\mathbf{F}$.
\end{proposition}

\begin{proof}
Invoking the triangle inequality, and since $\mathbf{U}\mathbf{S}\mathbf{U}^T$ is the optimal rank $d$ approximation of $\mathbf{F}$, we can write
\begin{align*}
\norm{\mathbf{F}-\mathbf{U}\mathbf{S}\mathbf{U}^T} &\leq \norm{\mathbf{F}-\mathbf{U}_c\mathbf{S}\mathbf{U}_c^T}\\
&\leq \norm{\mathbf{F}-\mathbf{U}\mathbf{S}\mathbf{U}^T} + \norm{\mathbf{U}\mathbf{S}\mathbf{U}^T-\mathbf{U}_c\mathbf{S}\mathbf{U}_c^T}.
\end{align*}
The latter term can again be bounded with the triangle inequality
\begin{align}\label{eq:proofkmeans}
&\norm{\mathbf{U}\mathbf{S}\mathbf{U}^T-\mathbf{U}_c\mathbf{S}\mathbf{U}_c^T}\notag\\
&\leq \norm{\mathbf{U}\mathbf{S}\mathbf{U}^T-\mathbf{U}_c\mathbf{S}\mathbf{U}^T} + \norm{\mathbf{U_c}\mathbf{S}\mathbf{U}^T-\mathbf{U_c}\mathbf{S}\mathbf{U_c}^T} \notag\\
&= \norm{(\mathbf{U}-\mathbf{U}_c)\mathbf{S}\mathbf{U}^T}+\norm{\mathbf{U}_c\mathbf{S}(\mathbf{U}^T-\mathbf{U}_c^T)}.
\end{align}
Observe that we can write $\mathbf{U}_c = \mathbf{C}\mathbf{U}$, where $\mathbf{C} \in \mathbb{R}^{n \times n}$ denotes a matrix with in each row entries equal to $1/n_i$ for nodes that are in the same bin, and zero otherwise. Here, $n_i$ denotes the size of a cluster $i$, i.e. $\mathbf{C}$ is a matrix that replaces each row in $\mathbf{U}$ by its cluster centroid.
It's easy to see that $\mathbf{C}$ is an orthogonal projection matrix, in the sense that $\mathbf{C}^2=\mathbf{C}$ and $\mathbf{C}=\mathbf{C}^T$. As such, for any matrix $\mathbf{B}$  it holds that  $\norm{\mathbf{C}\mathbf{B}}\leq\norm{\mathbf{B}}$.
Applying to the second term in (\ref{eq:proofkmeans}) yields
\begin{align*}
&\norm{(\mathbf{U}-\mathbf{U}_c)\mathbf{S}\mathbf{U}^T}+\norm{\mathbf{U}_c\mathbf{S}(\mathbf{U}^T-\mathbf{U}_c^T)}\\ 
&\leq \norm{(\mathbf{U}-\mathbf{U}_c)\mathbf{S}\mathbf{U}^T}+\norm{\mathbf{U}\mathbf{S}(\mathbf{U}^T-\mathbf{U}_c^T)}\\
&\leq 2 \sqrt{\textstyle \sum_{i=1}^{d}|\lambda_i|} \cdot \norm{\mathbf{U}-\mathbf{U_c}}.
\end{align*}
The result follows by noting that $\norm{\mathbf{U}-\mathbf{U_c}}$ is the k-means objective function and applying the bound in expectation for the kmeans++ algorithm \cite{kmeanspp}.
\end{proof}

The bound from Proposition~\ref{thm:bound} gives insight in how the dimension $d$ and number of bins $k$ affect the block-approximation (\ref{fsimap}).
Increasing $k$ will benefit the approximation since $\overline{\mathbf{F}}$ becomes closer to the optimal rank $d$ approximation.
However, for fixed $k$, increasing $d$ does \emph{not} always benefit the block-approximation.
Indeed, as the clustering approximation gets worse with increasing dimension $d$ \cite{kmeansSpeed}, the effect on the overall block-approximation could potentially be detrimental, which is confirmed in practice.
As a practical guideline, we advise to keep $d$ small, and selecting a high $k$ while maintaining tractable computational complexity (Section~\ref{sec:ort}).

Since $\mathbf{F}$ is often dense, directly calculating $\mathbf{F}$ and performing a top $d$ eigendecomposition is not scalable both memory and timewise.
Instead, for \emph{higher-order proximity} matrices, i.e. matrices that are expressed as polynomials of $\mathbf{\hat{A}}$ with positive coefficients, one can directly use a (fast) eigendecomposition of $\mathbf{\hat{A}}$ (which is typically sparse) to get the top $d$ eigendecomposition \cite{HIGHAM20035,Zhang2018ArbitraryOrderPP}.
Indeed, if $\mathbf{F}$ is of the form
\begin{equation*}
\mathbf{F}_\text{P} \triangleq \text{poly}(\mathbf{\hat{A}}) = q_1\mathbf{\hat{A}}+\ldots+q_n\mathbf{\hat{A}}^n\quad q_i \geq 0,
\end{equation*}
then it's trivial to see that if $\alpha$ is an eigenvalue with eigenvector $x_{\alpha}$ of $\mathbf{\hat{A}}$ then $\text{poly}(\alpha)$ will be an eigenvalue of $\mathbf{F}_\text{P}$ with the same eigenvector $x_{\alpha}$.
Hence eigenvalues are rescaled, and eigenvectors are preserved.
The only difficulty is that this rescaling does not preserve the ordering. More precisely, the top $d$ eigenvalues of $\mathbf{F}_\text{P}$ are in general not equal to the rescaled top $d$ eigenvalues of $\mathbf{\hat{A}}$.
To get the top $d$ eigenvalues of $\mathbf{F}_\text{P}$ one needs to calculate $l\geq d$ eigenvalues of $\mathbf{\hat{A}}$, where $l$ denotes the index of the $d$-th positive eigenvalue of $\mathbf{\hat{A}}$, when sorted according to absolute value.
This is true since $q_i \geq 0$ guarantees that the ordering is preserved only for the $\emph{positive}$ eigenvalues of $\mathbf{\hat{A}}$.
In practice $l$ is often not significantly larger than $d$. For example, for sufficiently large Erdos-Renyi graphs $l \approx 2d$ due to the semicircle law \cite{Erdossemi}.

\subsubsection{Resource Allocation Index (RAI) and Adamic-Adar (AA)}\label{sec:rai}
One is not limited to polynomials of $\mathbf{\hat{A}}$ for scalable block-approximation.
For other practical feature matrices $\mathbf{F}$, we can still utilize an eigendecomposition of $\mathbf{\hat{A}}$ to get a fast block-approximation of $\mathbf{F}$.
One such matrix often used in the complex networks community is the so-called \emph{Resource Allocation Index} (RAI) \cite{ZhouResourceAllocation}.
The RAI defines a similarity score $r$ between two nodes $u$ and $v$ as $r_{uv} = \textstyle\sum_{k \in \mathcal{N}(u) \cap \mathcal{N}(v)} 1/d_k$.
The induced matrix $\mathbf{F}_\text{RAI} \triangleq [r_{uv}]$ can be written in terms of the adjacency matrix $\mathbf{\hat{A}}$:
\begin{equation}\label{eq:rai}
\mathbf{F}_\text{RAI} = \mathbf{\hat{A}}\mathbf{D}^{-1}\mathbf{\hat{A}},
\end{equation}
where $\mathbf{D}$ is a diagonal matrix containing the degree of each node. Assume $\mathbf{\hat{A}}$ is connected, such that $d_k>0$ and (\ref{eq:rai}) is well-defined.
Given a top $d$ eigendecomposition of $\mathbf{\hat{A}} \approx \mathbf{V}\Lambda\mathbf{V}^T$, with orthonormal columns of $\mathbf{V} \in \mathbb{R}^{n \times d}$, one obtains a rank $d$ approximation of $\mathbf{F}_\text{RAI}$ as follows:
\begin{align*}
\mathbf{F}_\text{RAI}  &\approx \mathbf{V}(\Lambda\mathbf{V}^T\mathbf{D}^{-1}\mathbf{V}\Lambda)\mathbf{V}^T \\
 &= \widetilde{\mathbf{V}}\widetilde{\mathbf{V}}^T,
\end{align*}
where $\widetilde{\mathbf{V}} \triangleq\mathbf{V}(\Lambda\mathbf{V}^T\mathbf{D}^{-1}\mathbf{V}\Lambda)^{1/2} \in \mathbb{R}^{n\times d}$.
These expressions are well-defined, since the positive definiteness of $\mathbf{V}^T\mathbf{D}^{-1}\mathbf{V} \in \mathbb{R}^{d\times d}$ implies positive definiteness of $\Lambda\mathbf{V}^T\mathbf{D}^{-1}\mathbf{V}\Lambda \in \mathbb{R}^{d\times d}$, hence the principal root exists and is unique.
To obtain a block-approximation $\overline{\mathbf{F}}_\text{RAI}$, simply cluster the rows of $\widetilde{\mathbf{V}}$ into bins and replace the rows by centroids.

The \emph{Adamic-Adar Index} (AA) \cite{Adamic} can be block-approximated in a very similar fashion.
It is defined similarly as (\ref{eq:rai}), by substituting $d_k$ by $\log(d_k)$.
Nodes with $d_k=1$ lead to an ill-defined $\mathbf{D}^{-1}$ matrix, but since they never occur as a common neighbor of two other nodes, these nodes are omitted in the calculations.

\subsubsection{Preferential Attachment (PA)}\label{sec:prefat}
The preferential attachment feature matrix \cite{Grover:2016ex,Parisi} is defined as $\mathbf{F}_\text{PA} \triangleq [|\mathcal{N}(i)| \cdot |\mathcal{N}(j)|]$, i.e., the matrix induced  by the product of the degrees.
By definition $\mathbf{F}_\text{PA}$ is already rank one; it is the outer product of a vector of degrees with itself.
For reasons discussed in the proof of Theorem~\ref{lem:redvar}, the natural way to block-approximate (in this case, exactly) $\overline{\mathbf{F}}_\text{PA}=\mathbf{F}_\text{PA}$ is by considering the unique degrees in the network.

\subsubsection{Cross/Skeleton decompositions}\label{sec:lowrankmethods}
Alternatively, general methods from the vast literature on scalable low-rank approximations \cite{GOREINOV, Achlioptas, Markovsky, KumarSchneider, Indyk2019SampleOptimalLA} can be utilized.
The most scalable methods (Cross/Skeleton decompositions) essentially sample rows and columns to obtain a low-rank decomposition.
We did not utilize these methods. Instead, we restricted ourselves to the feature matrices defined above, for which a fast top $d$ eigendecomposition of $\mathbf{\hat{A}}$ was sufficient to obtain qualitative block-approximations.


\subsection{Overall running time}\label{sec:ort}

\ptitle{Eigendecomposition \& k-means}
Computing the top $l$ eigendecomposition of $\mathbf{\hat{A}}$ is efficient for sparse symmetric matrices using iterative methods \cite{Lehoucq1996DeflationTF, Stewart}, scaling linearly with $n$ for a fixed number of iterations.
Moreover, this is only computed once. Running the k-means++ algorithm for $t$ iterations has time complexity $O(tnkd)$ \cite{kmeanspp}. There are known instances \cite{kmeansSpeed} for which $t = 2^{\Theta(\sqrt{n})}$ until convergence, but the $O(\log(k))$ approximation ratio in expectation (Proposition 1) is valid even after the initialization of the clustering ($t=1$). Hence, overall running time is linear in $n$ for fixed $t, k$ and $d$.

\ptitle{Optimizing the reduced Lagrange dual function}
For sparse graphs $(m=O(n))$, Theorem~\ref{lem:redvar} shows that the final computational step is to solve an unconstrained convex optimization problem with $O(\sqrt{kn})$ variables.
The computational complexity for computing the gradient as well as the Hessian is $O(kn)$.
Space complexity for storing the gradient is  $O(\sqrt{kn})$ and $O(kn)$ for the Hessian, which roughly (for small $k$) equals the space complexity of storing the graph in sparse representation.
In Section~\ref{sec:runtime} we compare three different optimization strategies, and conclude that L-BFGS \cite{Liu1989} is particularly well-suited for this objective function.
This has been observed before, as L-BFGS has been described as the ``algorithm of choice'' for fitting log-linear (i.e., maximum entropy) models \cite{Malouf2002,andrew2007scalable}.

\subsection{Combining multiple feature matrices}\label{sec:combining}
Theorem~\ref{lem:redvar} is formulated for the case of only one block-approximated matrix $\overline{\mathbf{F}}$.
Combining multiple block-approximations $\overline{\mathbf{F}}_1,\ldots,\overline{\mathbf{F}}_{\gamma}$ can be done by considering the \emph{greatest lower bound} of the node binning.
Each $\overline{\mathbf{F}}_i$ defines a partition $B_i$ on the set of nodes.
A partition $B_i$ is a refinement of a partition $B_j$ if each element of $B_i$ is a subset of some element in $B_j$.
This relation $B_i \leq B_j$ defines a partial order \cite{birkhoff1967lattice} and the set of all partitions form a lattice.
A given set of partitions  $\{B_1, \ldots, B_{\gamma}\}$ thus has a greatest lower bound $b \leq B_i$.
Theorem~\ref{lem:redvar} still holds for multiple $\overline{\mathbf{F}}_i$ matrices, if one replaces $k$ with $|b|$ (and adds an additive term of $\gamma$ instead of one).
Assuming an equal number of bins $k$ for each $\overline{\mathbf{F}}_i$, worst-case this amounts to a lower bound with $|b| = \min\{k^{\gamma},n\}$ bins.
However, for a limited number of global constraints $\gamma$ and small $k$, we already obtain qualitative and scalable results in practice (Section~\ref{sec:exp}).

\section{Evaluation}
\label{sec:exp}

\begin{table}
	\caption{Summary of the datasets.}
	\label{table:datasets}
	\scalebox{1}{
		\begin{tabular}{lccc} \\ \hline
			Dataset 		& Category & $|V|$ & $|E|$ \\ \hline
			{StudentDB} \cite{goethals2010}		& Relational & $395$   	& $3,423$	\\
			{Facebook} \cite{leskovec2015snap}	& Social & $4,039$	& $88,234$	\\
			{PPI} \cite{breitkreutz2007biogrid}    & Biological & $3,852$ 	& $37,841$ \\
			{Wikipedia} \cite{mahoney2011large}	& Language & $4,777$ 	& $92,295$ \\
			{GR-QC} \cite{leskovec2015snap}	& Collaboration & $4,158$ 	& $26,844$	\\
			{BlogCatalog} \cite{zafarani2009social}	& Social & $10,312$ & $333,983$	\\
			{YouTube} \cite{mislove-2007-socialnetworks}  & Social & $1,138,499$ & $2,990,443$ \\
			{Flickr} \cite{leskovec2015snap}    & Social & $80,513$ & $11,799,764$ \\
			{DBLP} \cite{leskovec2015snap}	& Collaboration & $317,080$ & $1,049,866$	\\ \hline
		\end{tabular}
	} 
\end{table}

We test the performance of our ERG models on an important downstream task: link prediction (Section~\ref{sec:lp}). To ensure reproducibility for the link prediction experiment, we utilize the EvalNE library \cite{mara2019evalne}. Model implementations, as well as customized configuration files describing the experiments, are publicly available.\footnote{bitbucket.org/ghentdatascience/maxent-public}
Section~\ref{sec:gof} evaluates our proposed models on a social network using a goodness-of-fit approach.
Detailed runtime experiments and optimization strategies are discussed in Section~\ref{sec:runtime}.
Datasets are listed in Table \ref{table:datasets}.
Experiments were conducted on a Linux server with 256GB of RAM.

\begin{table*}
	\centering
	\caption{Average AUC for link prediction over three experiment repeats with different train/test splits for all methods. We use $\times$ to indicate that a method did not finish within a pre-set time of 4 hours. Best results for each dataset are highlighted in bold. }
	\label{table:lp}
	\scalebox{1}{
		\begin{tabular}{lccccccccc}
			\toprule{}
			& StudentDB & Facebook & PPI & Wiki & GR-QC & BlogCatalog & Flickr & YouTube & DBLP\\
			\midrule
			CN & 0.4101 & 0.9792 & 0.7737 & 0.8427 & 0.8602 & 0.9343 & 0.9379 & 0.5831 & 0.8127 \\
			JC & 0.4101 & 0.9754 & 0.7613 & 0.4954 & 0.8598 & 0.8045 & 0.9316 & 0.5831 & 0.8127 \\
			AA & 0.4101 & 0.9807 & 0.7764 & 0.8681 & 0.8604 & 0.9396 & 0.9383 & 0.5831 & 0.8127 \\
			PA & 0.9202 & 0.8392 & 0.9022 & 0.9175 & 0.8311 & 0.9638 & 0.9676 & 0.9913 & 0.8866 \\
			RAI & 0.4101 & 0.9813 & 0.776 & 0.8753 & 0.8603 & 0.9399 & 0.9376 & 0.5831 & 0.8127 \\ \hline
			DeepWalk & 0.8865 & 0.9878 & 0.8867 & 0.8903 & \textbf{0.9627} & 0.9393 & 0.9772 & $\times$ & $\times$ \\
			Node2vec & 0.9144 & \textbf{0.993} & 0.8552 & 0.8923 & 0.9593 & 0.922 & $\times$ & $\times$ & $\times$ \\
			Struc2vec & 0.92 & 0.8309 & 0.9006 & 0.9167 & 0.8215 & 0.96 & $\times$ & $\times$ & $\times$ \\
			Role2vec & 0.8653 & 0.9753 & 0.7979 & 0.7398 & 0.9386 & 0.8066 & $\times$ & $\times$ & $\times$ \\
			LINE & 0.9259 & 0.9875 & 0.8826 & 0.8628 & 0.9448 & 0.947 & $\times$ & $\times$ & \textbf{0.9026} \\
			SDNE & 0.9695 & 0.9647 & 0.8885 & 0.9147 & 0.9066 & 0.9382 & $\times$ & $\times$ & $\times$ \\
			CNE & 0.8227 & 0.9082 & 0.8485 & 0.8611 & 0.8216 & 0.9193 & $\times$ & $\times$ & $\times$ \\
			AROPE & \textbf{0.9774} & 0.9863 & 0.899 & 0.9112 & 0.9191 & 0.9617 & \textbf{0.9825} & 0.9103 & 0.8757 \\ \hline
			MaxEnt ($k=5$) & 0.9542 & 0.8853 & 0.9018 & 0.9179 & 0.8365 & \textbf{0.9644} & 0.9699 & \textbf{0.9919} & 0.8863 \\
			MaxEnt ($k=100$) & 0.9626 & 0.9406 & 0.9026 & 0.9178 & 0.8752 & 0.9638 & 0.9699 & 0.9669 & 0.8857 \\
			MaxEnt (full) & 0.9604 & 0.9694 & \textbf{0.9097} & \textbf{0.9182} & 0.9342 & $\times$ & $\times$ & $\times$ & $\times$ \\
			\bottomrule
	\end{tabular}}
\end{table*}

\subsection{Link Prediction}
\label{sec:lp}
In link prediction the aim is to identify missing links from a given network. In this task, we randomly remove 50\% of the edges such that the remaining network is still connected. The reduced network is used for training, the removed edges are used for testing.
We compare with eight state-of-the-art network embedding methods (Deepwalk \cite{perozzi2014deepwalk}, Node2vec \cite{grover2016node2vec}, Struc2vec \cite{ribeiro2017struc2vec}, Role2vec \cite{Ahmed18Role2vec}, LINE \cite{tang2015line}, SDNE \cite{wang2016sdne}, CNE \cite{kang2018cne} and AROPE \cite{Zhang2018ArbitraryOrderPP}) and five common heuristics: common neighbours (CN), Adamic-Adar index (AA), Jaccard coefficient (JC), preferential attachment (PA) and Resource Allocation Index (RAI).
Method descriptions, parameter tuning and further details on the experimental setup are discussed in Appendix B.

\ptitle{MaxEnt (full)} First, we evaluate an exact model according to (\ref{eq:maxentGeneral}) with constraints on node degrees, and global constraints on $\mathbf{F}_\text{CN}$, $\mathbf{F}_\text{RAI}$ and $\mathbf{F}_\text{PA}$. This model is denoted as MaxEnt (full) in Table~\ref{table:lp}.

\ptitle{MaxEnt (blocked)} Secondly, two block-approximated models where both $\overline{\mathbf{F}}_\text{CN}$ (Section~\ref{sec:hop}) and $\overline{\mathbf{F}}_\text{RAI}$ (Section~\ref{sec:rai}) are binned into $k \in \{5,100\}$ bins and with $d=128$.
The matrix $\overline{\mathbf{F}}_\text{PA}$ is naturally binned by the unique degrees (Section~\ref{sec:prefat}). These models are denoted as MaxEnt ($k=5$) and MaxEnt ($k=100$) in Table~\ref{table:lp}.

In Table~\ref{table:lp} we present the average Area Under the ROC Curve (AUC) over three experiment repeats with different train/test splits for all methods.
MaxEnt (full) performs well on all datasets and is never far from the optimal value achieved across all methods. As expected, the method does not scale to large networks.
The block-approximated models ran on all networks and display competitive results.
Figure~\ref{fig:time1} shows total execution times on the datasets where all methods succesfully terminated within time.



\begin{figure}[tp]
	\captionsetup[subfloat]{farskip=5pt,captionskip=1pt}
	\centering
	{\scalebox{0.8}{\includegraphics[width=1\linewidth,trim=0cm 0cm 0cm 1.1cm,clip]{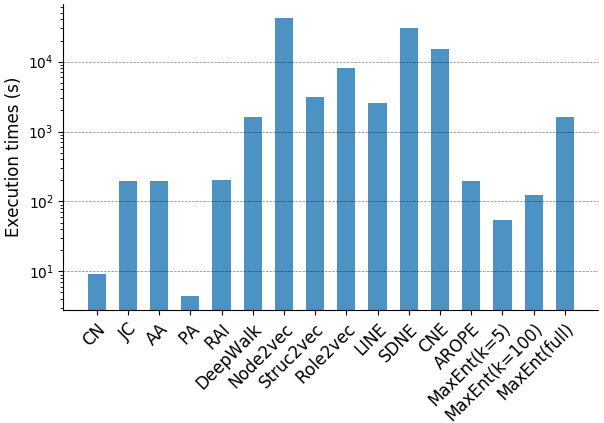}}}
	\caption{Total execution times for the link prediction task on the \emph{StudentDB}, \emph{Facebook}, \emph{PPI}, \emph{Wikipedia} and \emph{GR-QC} datasets.\label{fig:time1}}
\end{figure}

\begin{figure*}[tp]
	\captionsetup[subfloat]{farskip=5pt,captionskip=1pt}
	\centering
	\subfloat{\scalebox{0.8}{\includegraphics[width=0.28\linewidth,trim=0cm 0cm 0cm 1cm,clip]{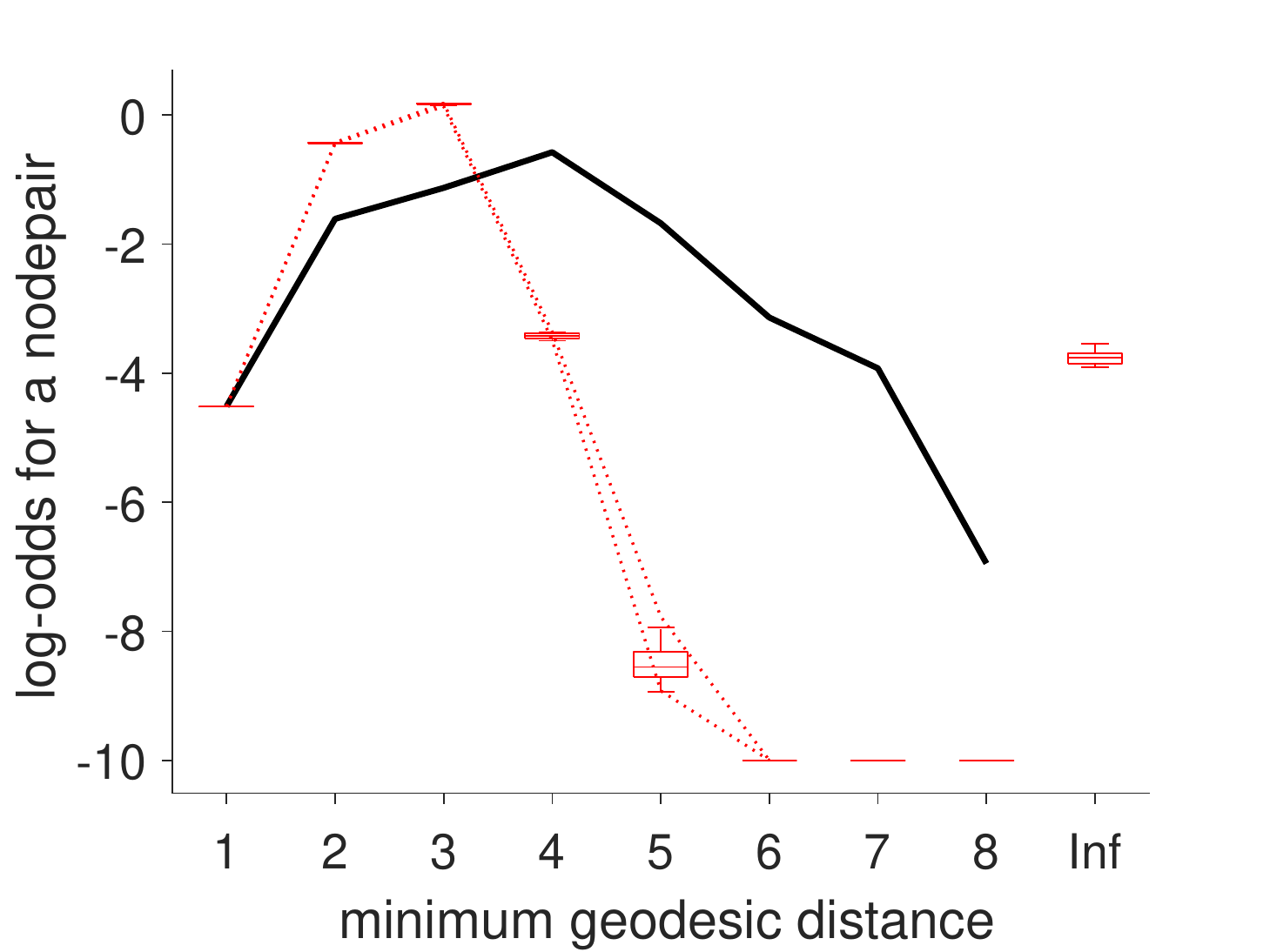}}}
	\subfloat{\scalebox{0.8}{\includegraphics[width=0.28\linewidth,trim=0cm 0cm 0cm 1cm,clip]{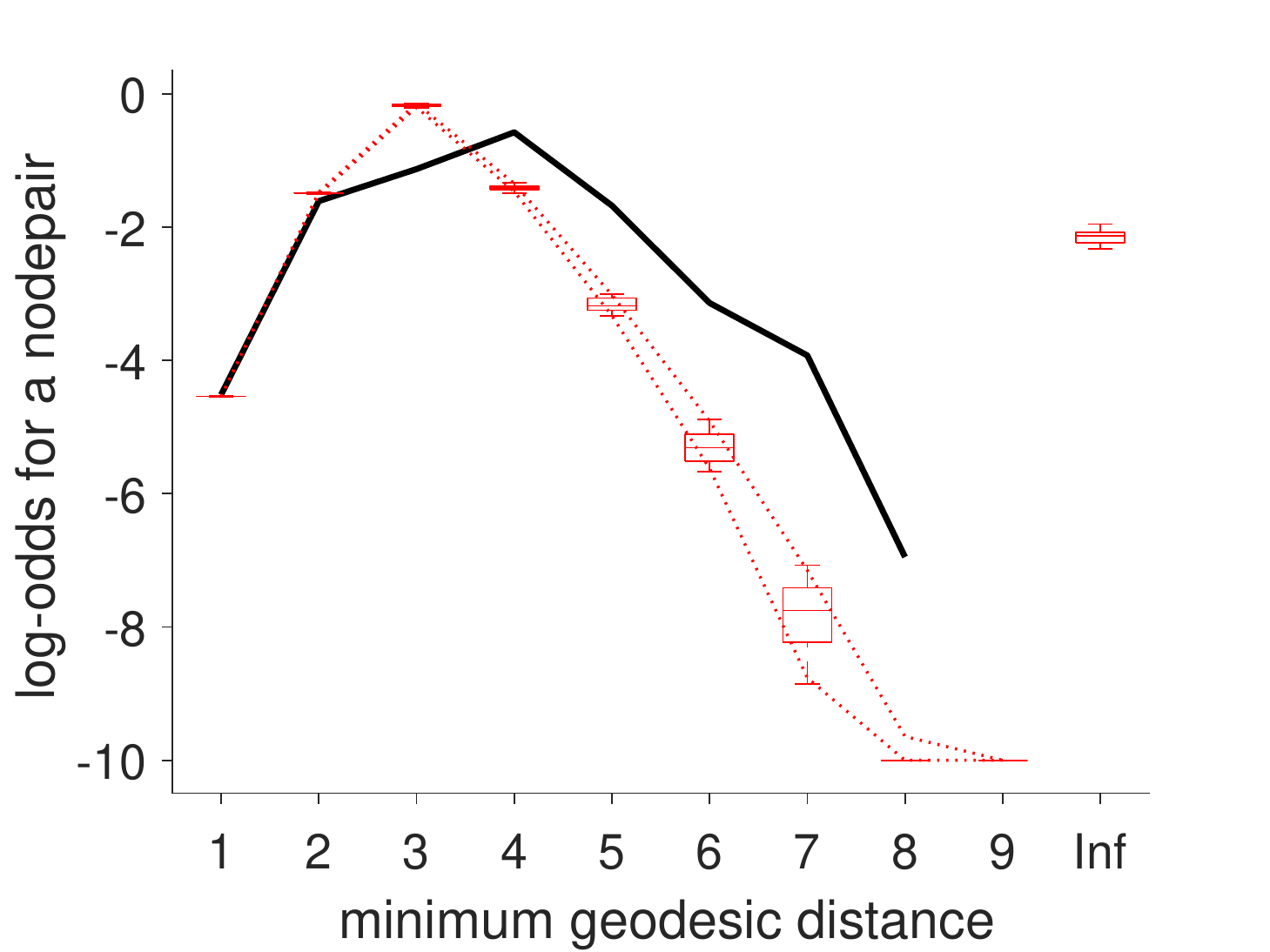}}}
	\subfloat{\scalebox{0.8}{\includegraphics[width=0.28\linewidth,trim=0cm 0cm 0cm 1cm,clip]{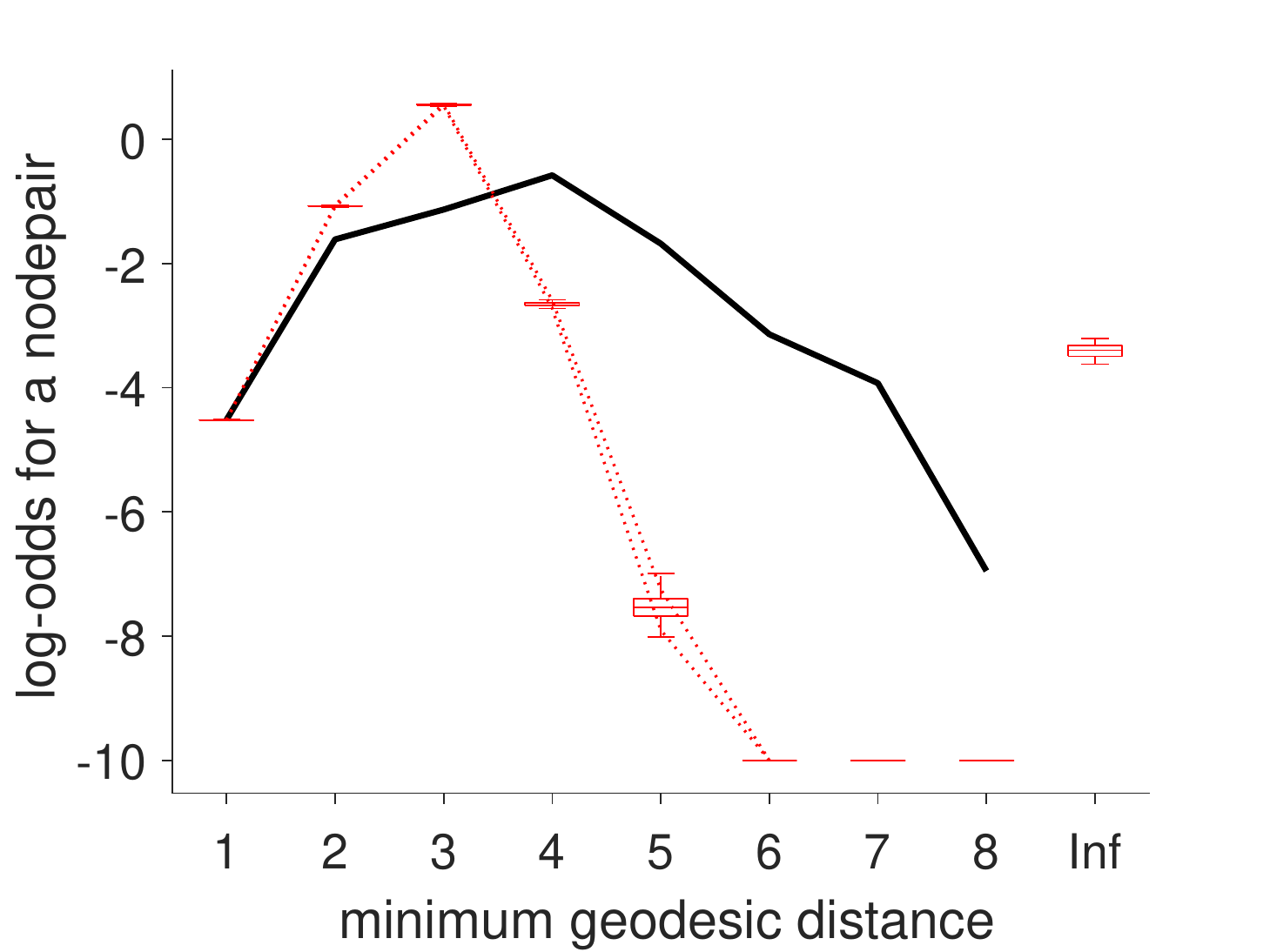}}}
	\subfloat{\scalebox{0.8}{\includegraphics[width=0.28\linewidth,trim=0cm 0cm 0cm 1cm,clip]{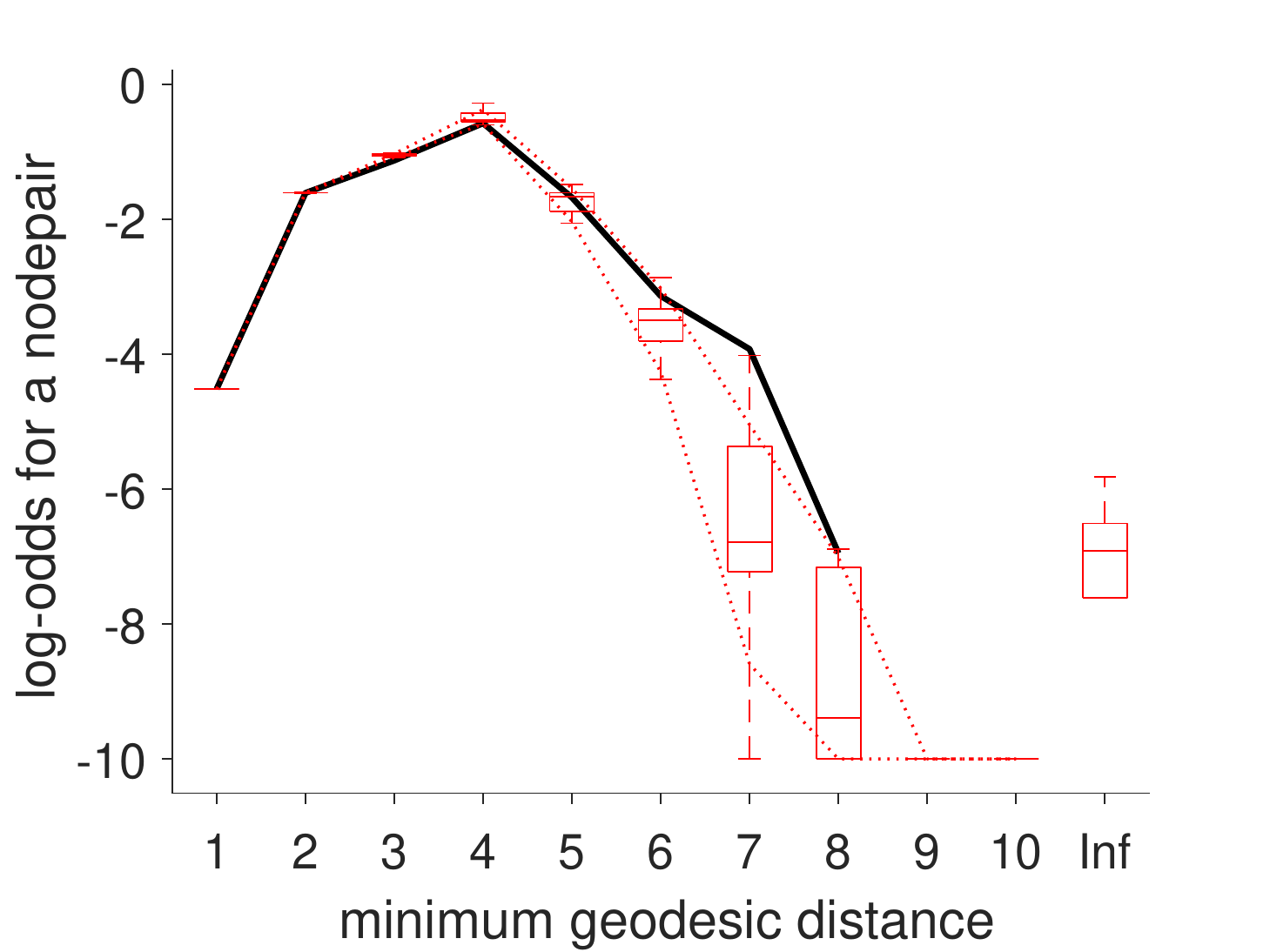}}}
	\\
	\subfloat{\scalebox{0.8}{\includegraphics[width=0.28\linewidth,trim=0cm 0cm 0cm 1.1cm,clip]{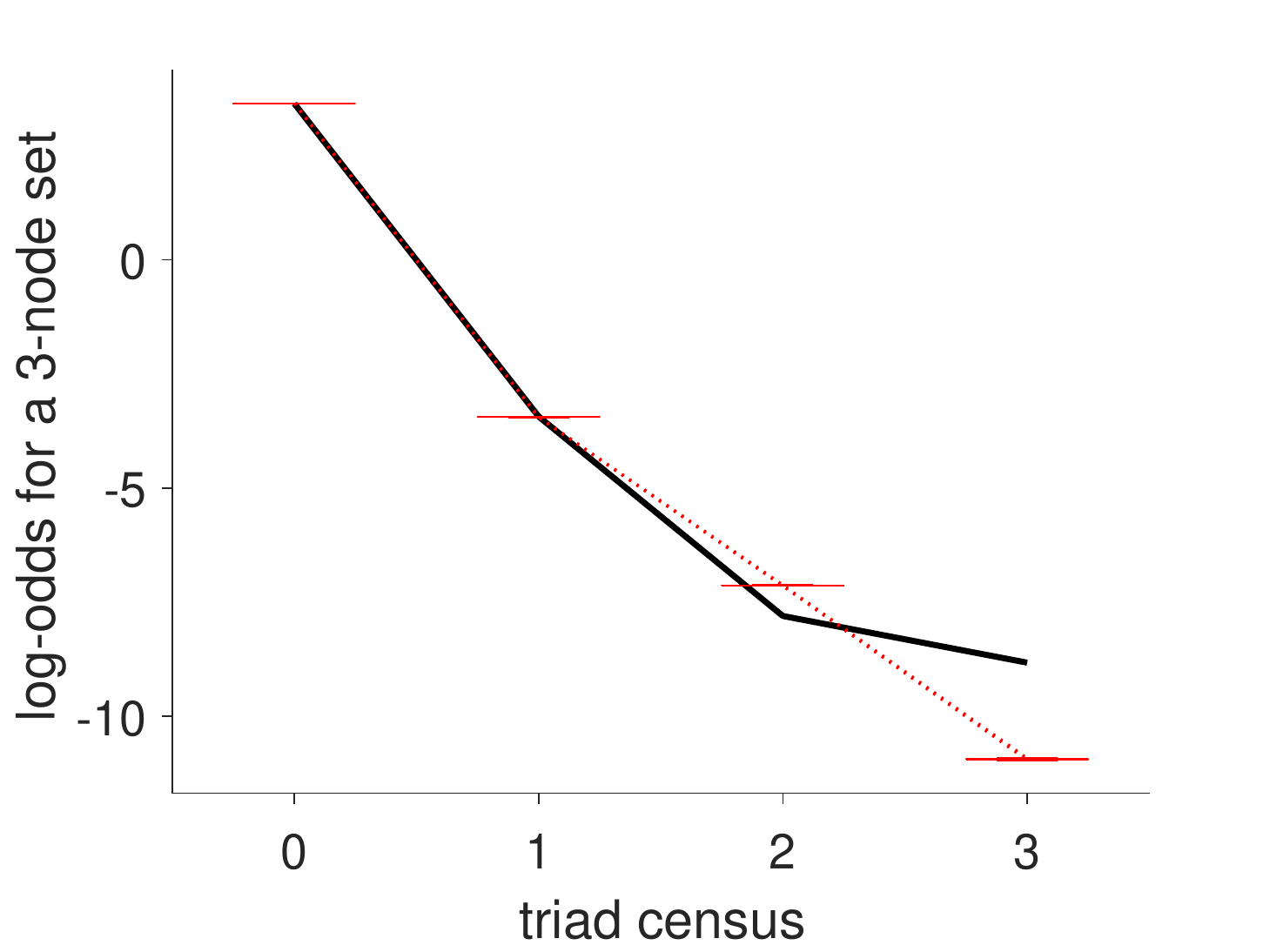}}}
	\subfloat{\scalebox{0.8}{\includegraphics[width=0.28\linewidth,trim=0cm 0cm 0cm 1.1cm,clip]{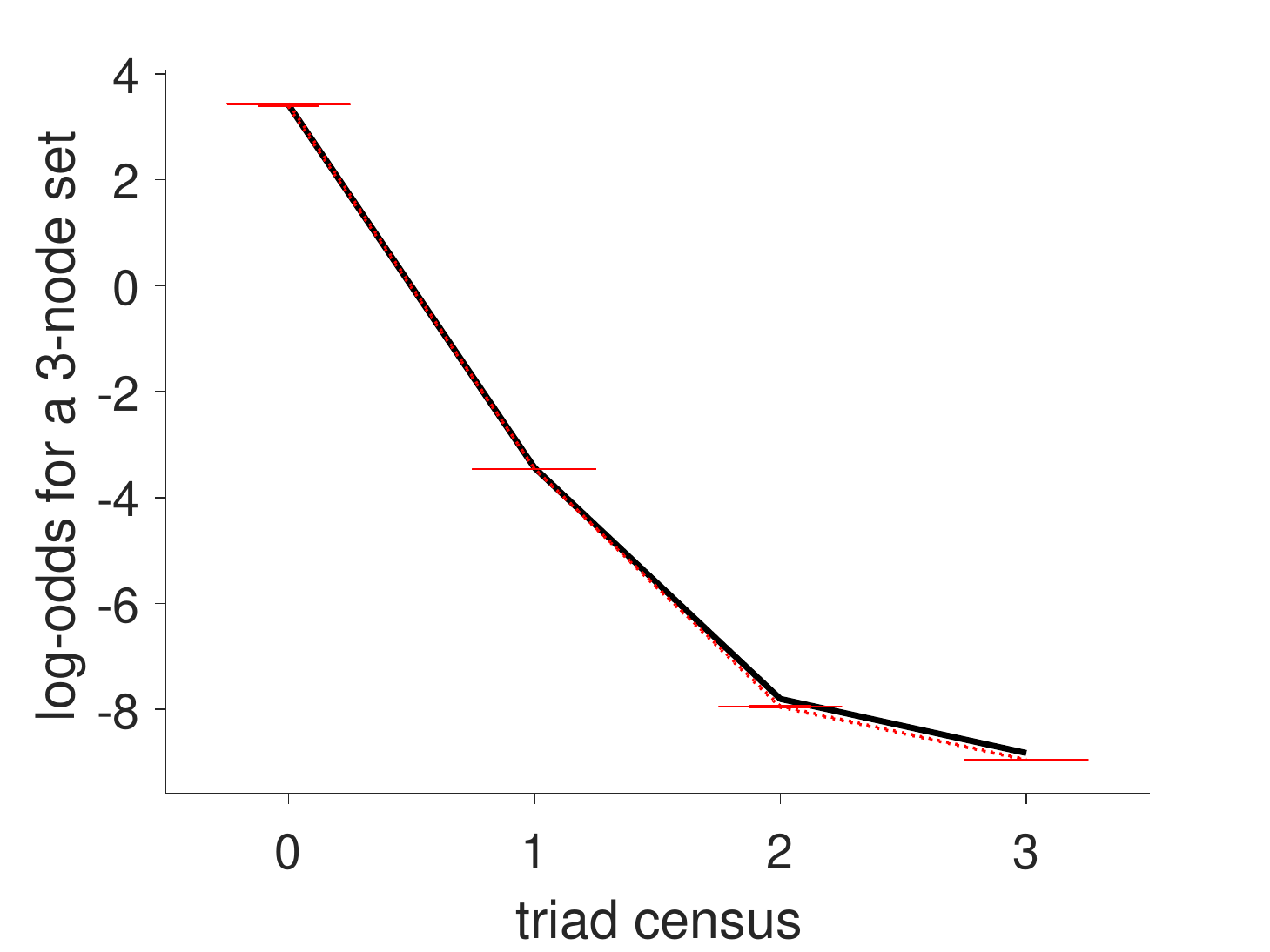}}}
	\subfloat{\scalebox{0.8}{\includegraphics[width=0.28\linewidth,trim=0cm 0cm 0cm 1.1cm,clip]{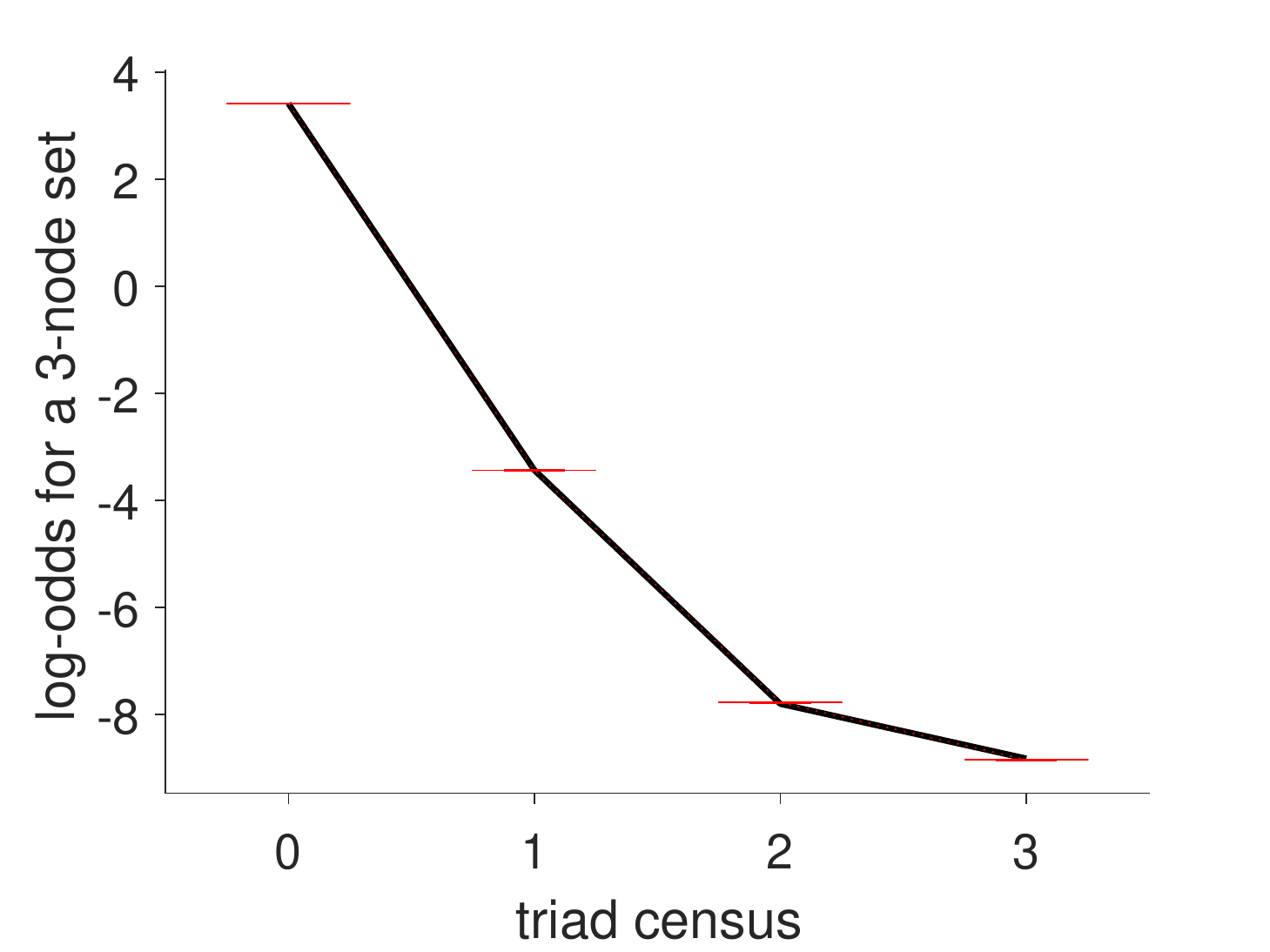}}}
	\subfloat{\scalebox{0.8}{\includegraphics[width=0.28\linewidth,trim=0cm 0cm 0cm 1.1cm,clip]{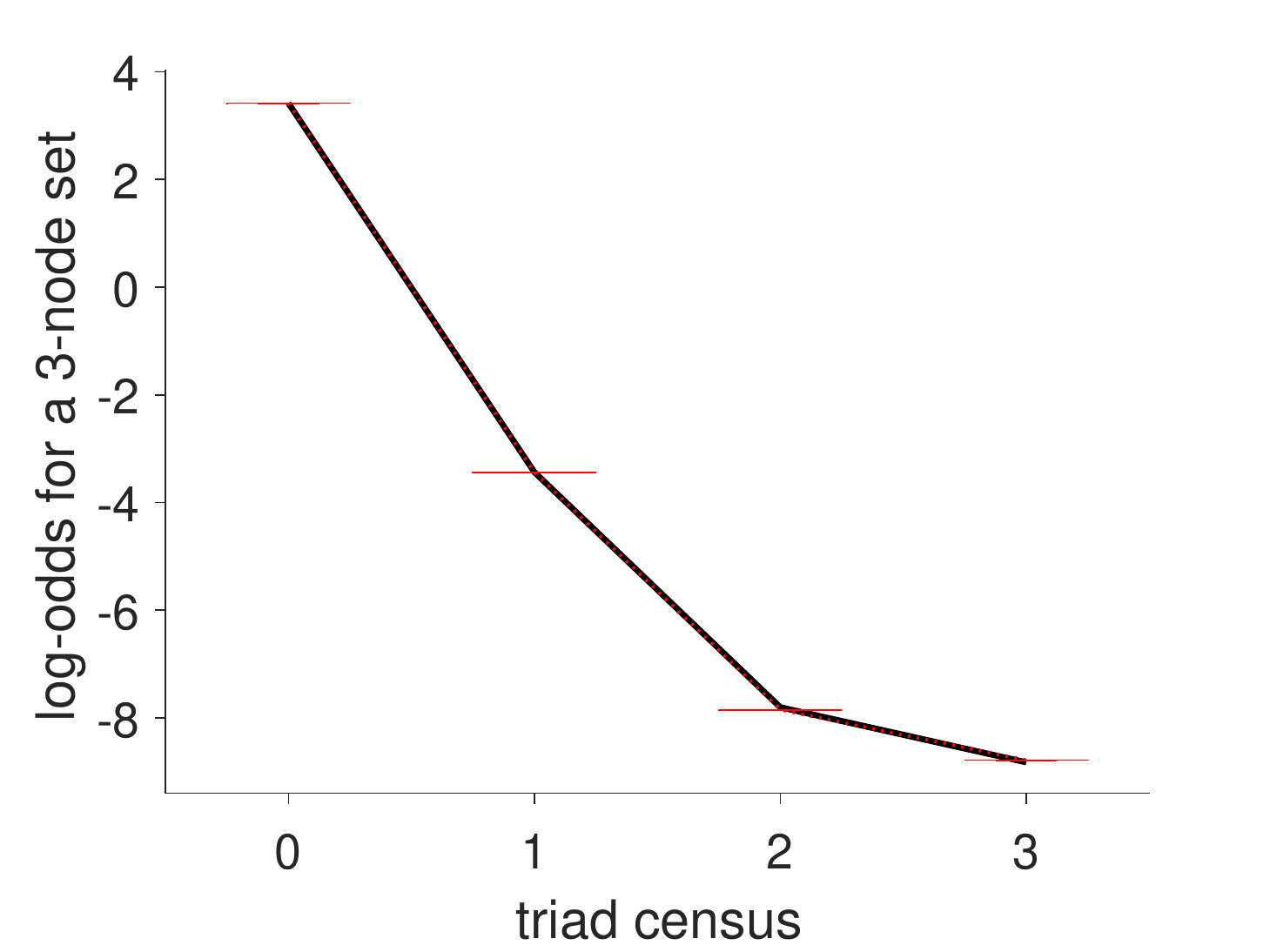}}}
	\\
	\renewcommand{\thesubfigure}{a}
	\subfloat[\small Chung-Lu model\label{fig:cl}]{\scalebox{0.75}{\includegraphics[width=0.28\linewidth,trim=0cm 0cm 0cm 0.5cm,clip]{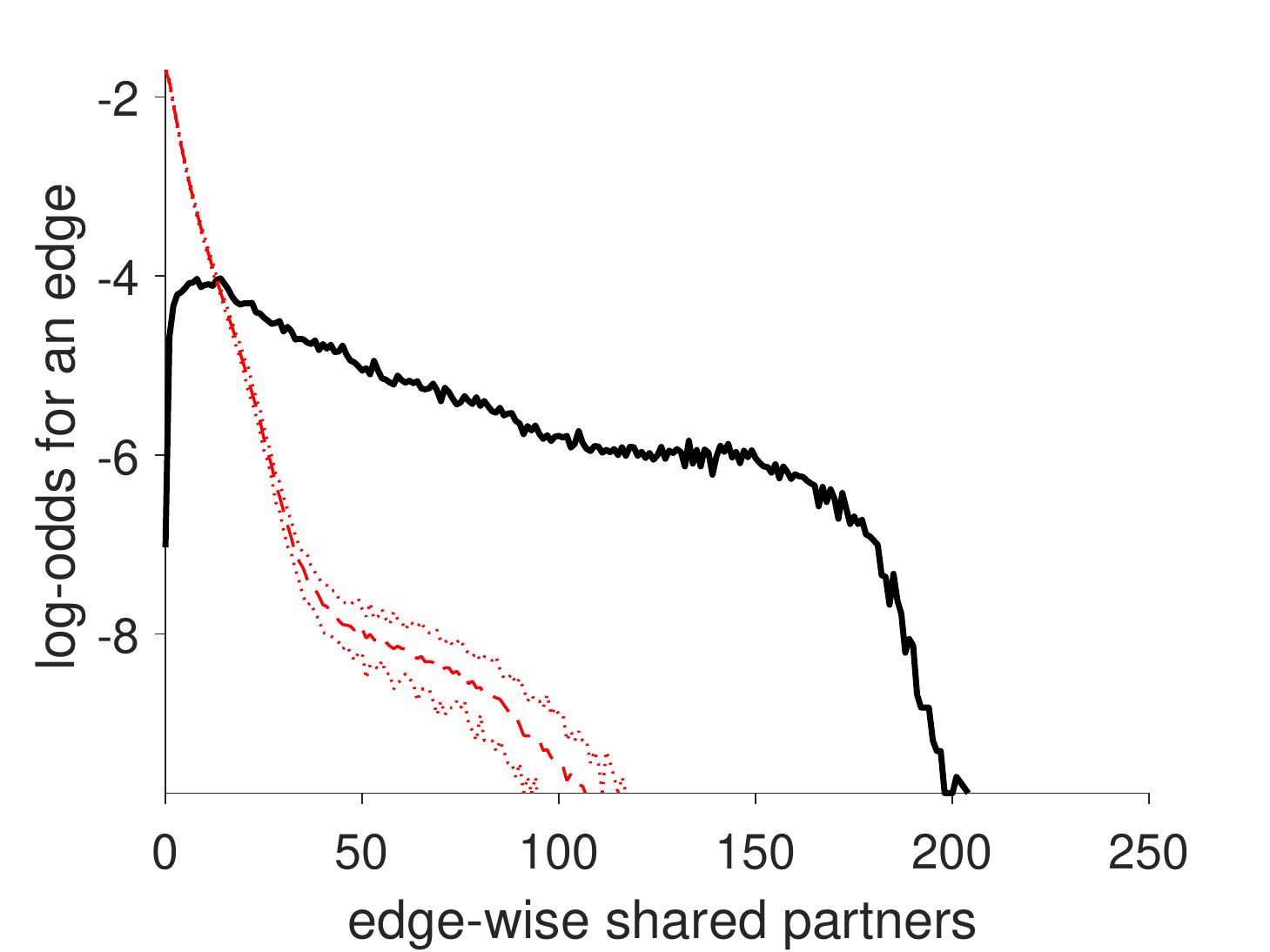}}}
	\renewcommand{\thesubfigure}{b}
	\subfloat[BTER model]{\scalebox{0.75}{\includegraphics[width=0.28\linewidth,trim=0cm 0cm 0cm 0.5cm,clip]{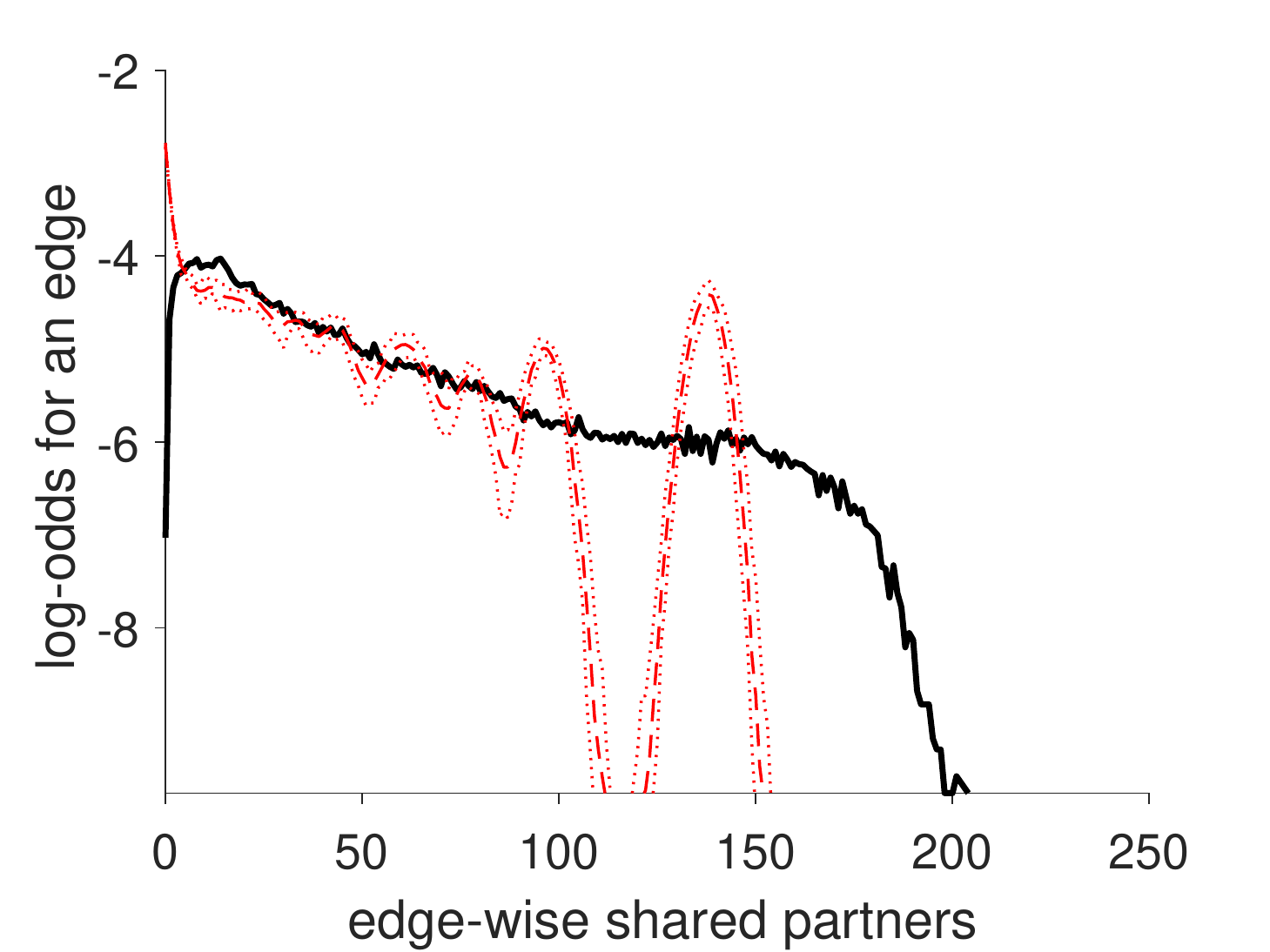}}}
	\renewcommand{\thesubfigure}{c}
	\subfloat[ME with $\overline{\mathbf{F}}_\text{PA}$ and $\overline{\mathbf{F}}_\text{RAI}$]{\scalebox{0.75}{\includegraphics[width=0.28\linewidth,trim=0cm 0cm 0cm 0.5cm,clip]{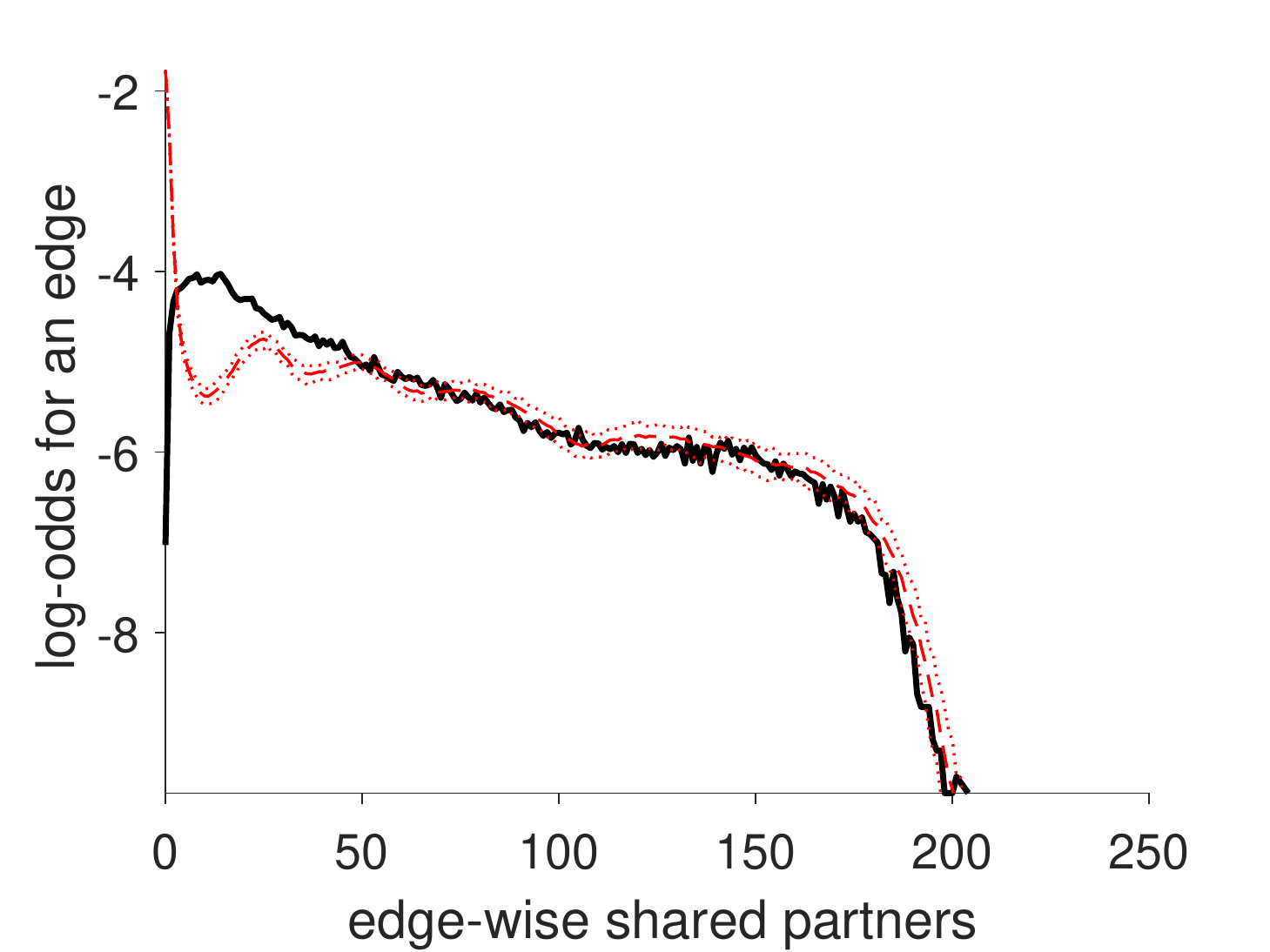}}}
	\renewcommand{\thesubfigure}{d}
	\subfloat[ME with $\mathbf{F} = \mathbf{\hat{A}}+\beta\mathbf{\hat{A}}^2$]{\scalebox{0.75}{\includegraphics[width=0.28\linewidth,trim=0cm 0cm 0cm 0.5cm,clip]{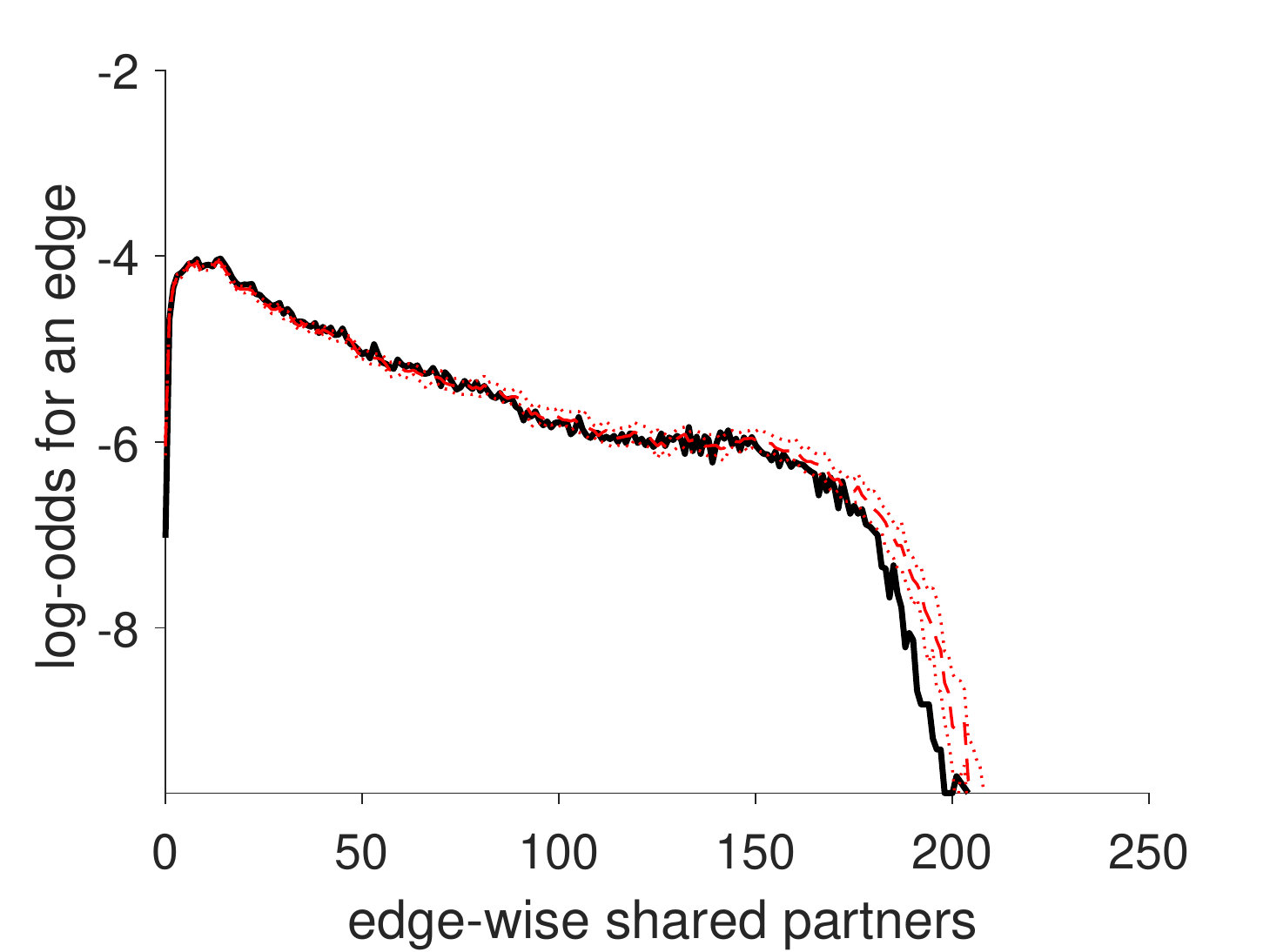}}}
	\caption{Goodness of fit plots on the \emph{Facebook} dataset for three statistics. The thick black lines are empirical distributions.
		The red lines are measurements from 50 randomly generated networks, according to four different independent edge models: (a) Chung-Lu (b) BTER model (c) MaxEnt with degrees and $\overline{\mathbf{F}}_\text{PA}$ and $\overline{\mathbf{F}}_\text{RAI}$ (block-approximated with $d=20$ and $k=100$) d) MaxEnt with degrees, and a constraint on $\mathbf{F} = \mathbf{\hat{A}}+\beta\mathbf{\hat{A}}^2$, with $\beta = 0.025$. Thin dotted red lines indicate 90\% confidence intervals.\label{fig:facebook}}
\end{figure*}

\subsection{Goodness-of-fit on a social network}
\label{sec:gof}
A standard visual approach for ERG model evaluation is to plot goodness-of-fits \cite{Goodreau,HunterGoodness}.
The idea is to compare a set of higher-order network statistics, preferably statistics that are not directly modeled, with a range of the same statistics obtained by simulating random graphs from the model.
We test on the \emph{Facebook} dataset, which is essentially a combination of social circles (communities), combined with individual ego-networks. We test on three statistics that have been previously proposed to evaluate social network models \cite{HunterGoodness}:

\squishlist
\item \textbf{Minimum geodesic distance}; the shortest path distances distribution between all pairs of nodes in the network. Unreachable node pairs get assigned the value `Inf'.
\item \textbf{Triad census}; the distribution of the number of edges between a set of three nodes.
\item \textbf{Edgewise shared partners}; let $\text{CN}_i$ denote the number of edges in the network that have exactly $0 \leq i \leq n-2$ common neighbors. This defines a distribution $\text{CN}_0/m, \text{CN}_1/m,\ldots, \text{CN}_{n-2}/m$.
\squishend

\begin{figure*}[tp]
\captionsetup[subfloat]{farskip=5pt,captionskip=1pt}
\centering
\subfloat[Runtime of different optimizers.\label{fig:time:opt}]{\scalebox{1}{\includegraphics[width=0.25\linewidth,trim=0cm 0cm 0cm 1.1cm,clip]{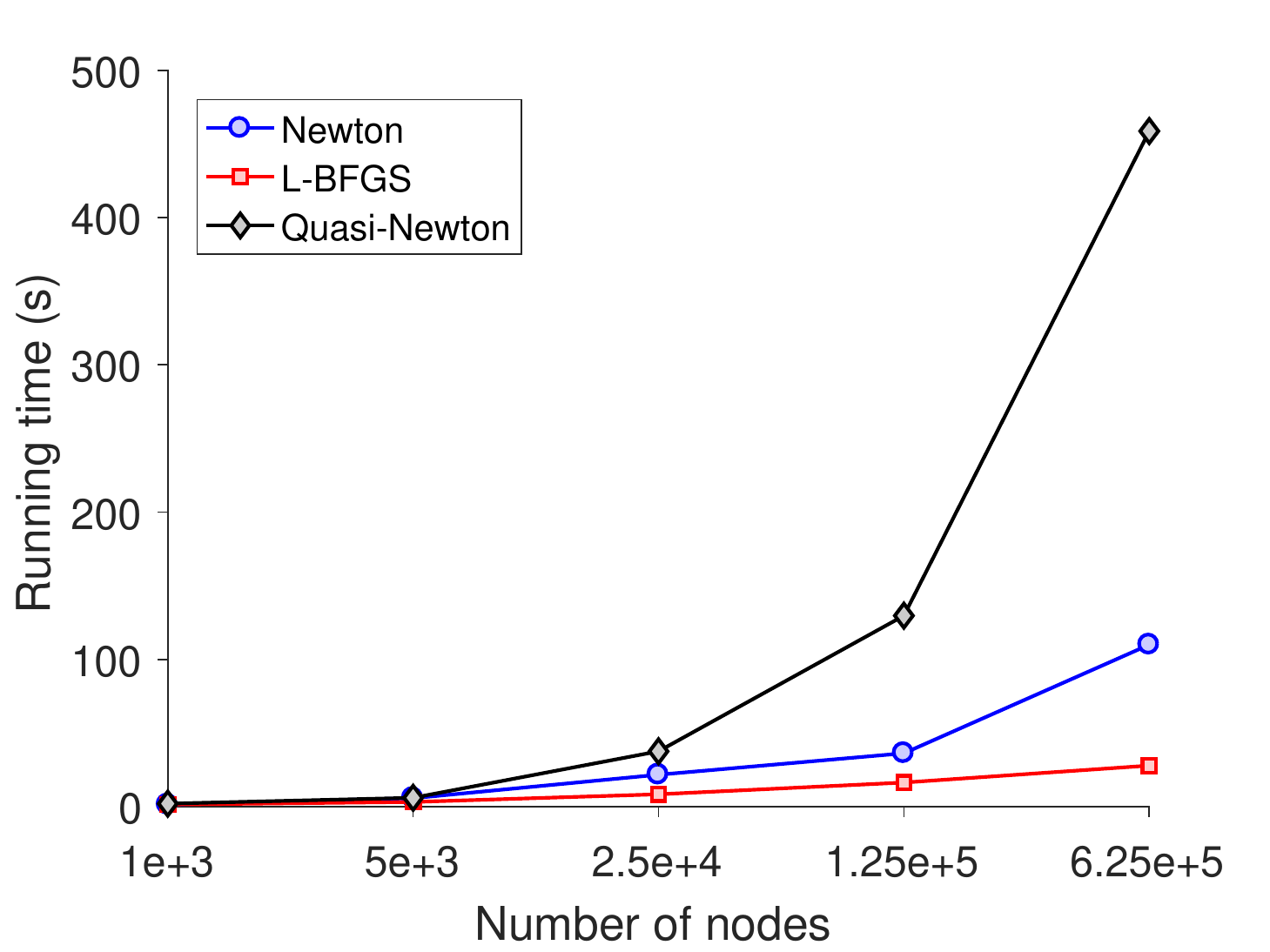}}}
\subfloat[Overall runtime.\label{fig:time:all}]{\scalebox{1.05}{\includegraphics[width=0.25\linewidth,trim=0cm 0cm 0cm 1.1cm,clip]{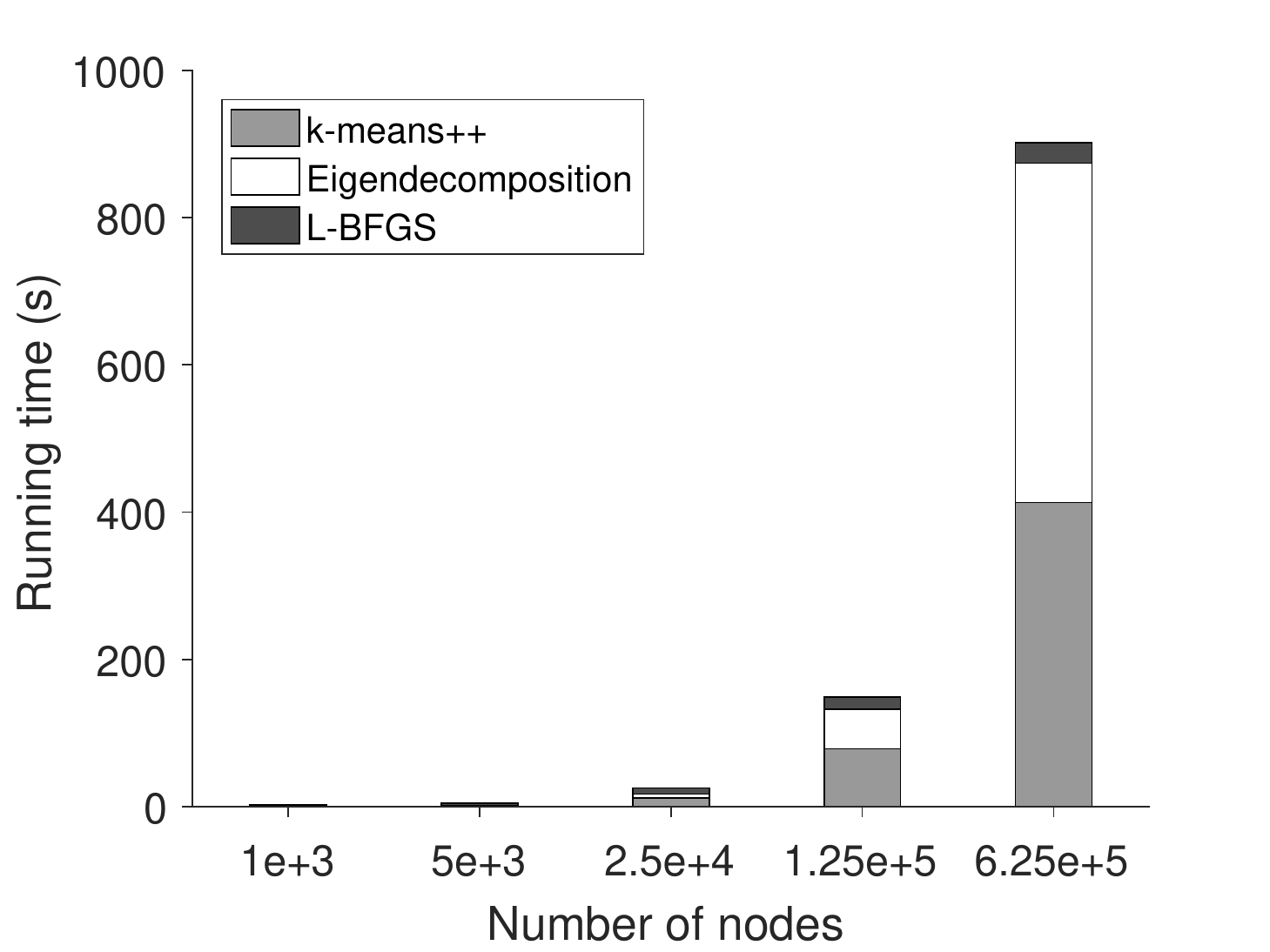}}}
\subfloat[Number of unique degrees.\label{fig:time:deg}]{\scalebox{1}{\includegraphics[width=0.25\linewidth,trim=0cm 0cm 0cm 1.1cm,clip]{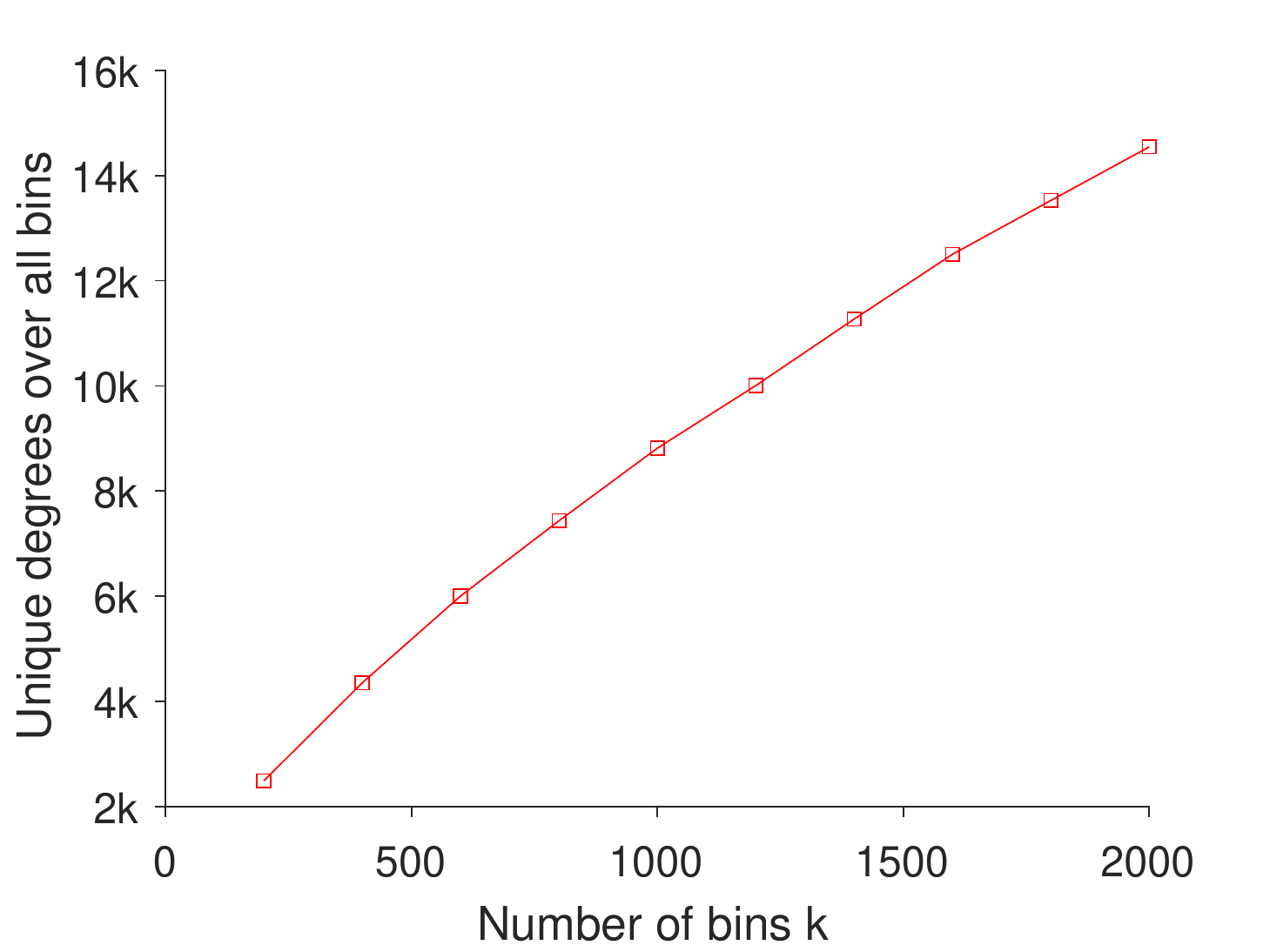}}}
\subfloat[Runtime L-BFGS vs bins $k$.\label{fig:time:lbfgs}]{\scalebox{1}{\includegraphics[width=0.25\linewidth,trim=0cm 0cm 0cm 1.1cm,clip]{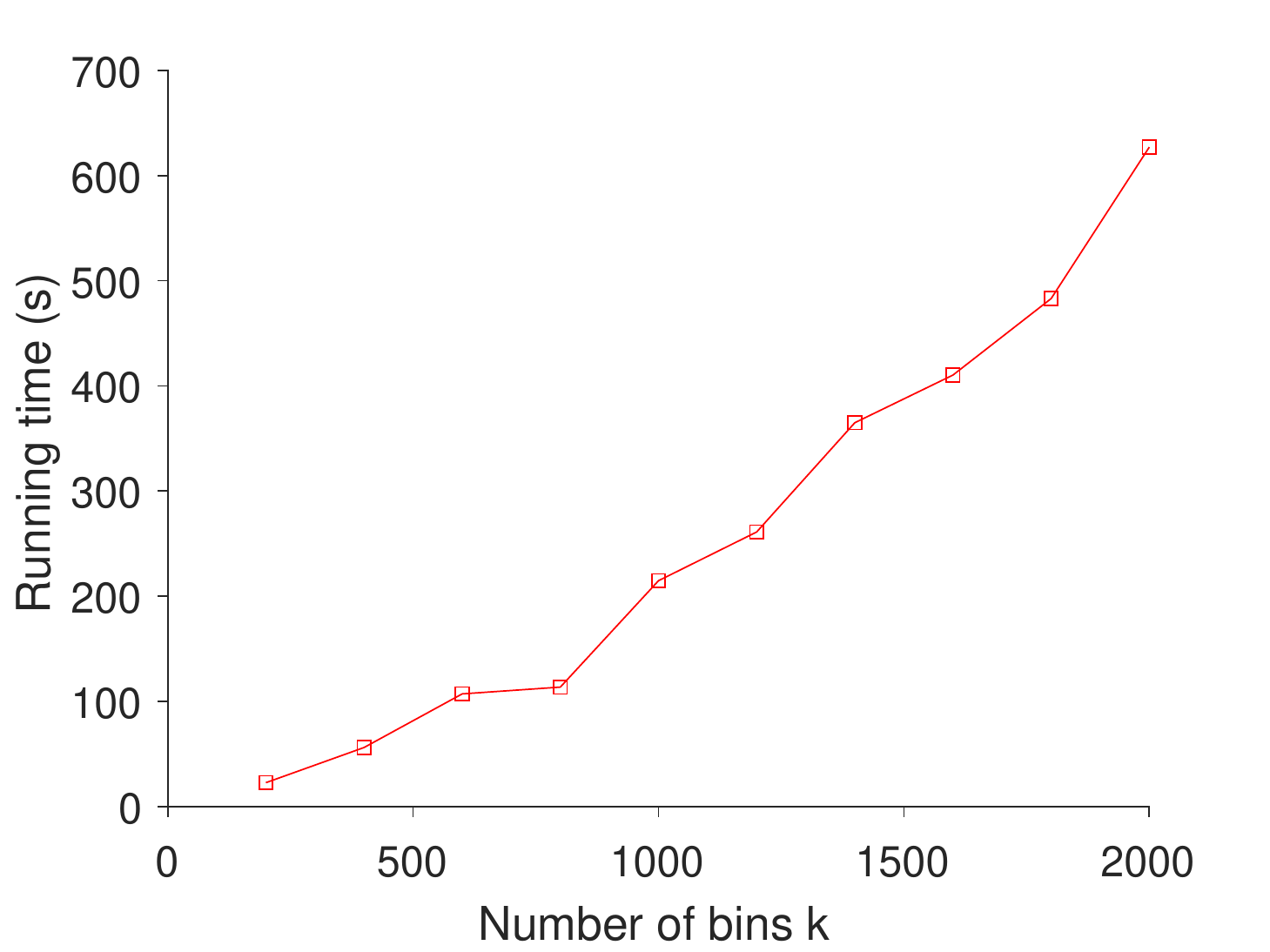}}}
\caption{Influence of graph size $n$ and bins $k$ on fitting times (norm gradient $<10^{-3}$) of a MaxEnt model with constraints on degrees and block-approximated $\overline{\mathbf{F}}_\text{CN}$.
In (a) and (b) we fix $k=100$ and let the graph size vary. In (c) and (d) we fix $n=10^5$ and let the bin size $k$ vary.\label{fig:time}}
\end{figure*} 

We fit an approximative MaxEnt model, with constraints on node degrees, block-approximated $\overline{\mathbf{F}}_\text{PA}$ (unique degrees) and $\overline{\mathbf{F}}_\text{RAI}$ ($d=20$, $k=100$).
Secondly, we fit an exact MaxEnt model with degree constraints and a polynomial constraint $\mathbf{F} = \mathbf{\hat{A}}+\beta\mathbf{\hat{A}}^2$, with $\beta=0.025$.
We compare with two other independent edge models.
The Chung-Lu model (CL) \cite{CL} assigns a probability of $d_id_j/\textstyle\sum_kd_k$ to each possible node pair $(i,j)$, where the probabilities are clipped to one if $d_id_j > \textstyle\sum_kd_k$.
Since the Chung-Lu model does not take into account other properties besides degrees, we do not expect this model to perform well on any of the aforementioned statistics.

To draw a more fair comparison, we compare with a model that takes community structure into account: the Block Two-Level Erdos-Renyi (BTER) model \cite{Kolda_2014,PhysRevE.85.056109}.
The BTER model is a well-known scalable generative network model that can be tuned to capture and preserve two fundamental properties: degree distribution and (local) clustering coefficients.

Figure~\ref{fig:facebook} shows the goodness-of-fit plots.
The Chung-Lu model fails to capture any form of community structure---which the dataset most surely has---resulting in poor performance on all statistics.
In contrast, both MaxEnt models with global constraints capture the triad census and edgewise shared partners distributions quite well.
The geodesic distances are less well captured by the block-approximated MaxEnt model.
The block-approximated MaxEnt model captures the community structure, but it still assigns a decent amount of probability mass to inter-community edges, reducing distances between nodes in different communities.

The BTER model performs well on the geodesic dinstances and the triad census, but interestingly, it shows adversial oscillatory behavior on the edgewise shared partners distribution.
Oscillatory behavior of statistics has been observed before in other generative models, e.g. it is known that Stochastic Kronecker Graphs have oscillating behavior of the degree distributions \cite{Seshadhri}.

The MaxEnt model fitted with the polymomial constraint is a very close fit to the original network, outperforming all other models (Figure~\ref{fig:facebook}d).
As such, it should be avoided for tasks like link prediction, but is useful to simulate networks that are close to the observed network.
We refer to Section~\ref{sec:featureleakage} for a discussion on label leakage.


\subsection{Detailed runtime testing}
\label{sec:runtime}

\ptitle{Influence of the graph size}
We generate synthetic Erdos-Renyi random graphs with varying number of nodes $n \in \{1, 5, 25, 125, 625\}\cdot10^3$ with edges $m \approx 10n$. We test the influence of the graph size $n$ on scalability, and compare three optimization strategies for optimizing (\ref{eq:lagrangian}).
We solve a MaxEnt model with constraints on degrees and a block-approximated $\overline{\mathbf{F}}_\text{CN}$ ($d=20$ and $k=100$).
We compare three different optimizers using the minFunc optimization package \cite{minFunc}:
i) full Newton's method
ii) a diagonal quasi-Newton's method
iii) L-BFGS, a well-known and popular quasi-Newton method for parameter estimation in machine learning \cite{Liu1989}.
All optimizers use the same Wolfe line-search criteria to ensure global convergence.
Figure~\ref{fig:time:opt} shows the time needed to reach a norm gradient tolerance of $10^{-3}$, showing the superior performance of L-BFGS.
Figure~\ref{fig:time:all} shows the overall time needed to block approximate $\overline{\mathbf{F}}_\text{CN}$, and then fitting the reduced model with L-BFGS.
K-means++ was repeated five times per instance, and each try was given a 100 iterations.
The overall fitting time is mostly dominated by the eigendecomposition of $\mathbf{\hat{A}}$ and the k-means clustering.

\ptitle{Influence of the binning}
Secondly, we generate 10 Erdos-Renyi graphs with fixed graph size $n=10^5$ and $m \approx 10n$. The number of bins $k$ are varied according to $200\leq k \leq 2000$, in steps of 200.
Figure~\ref{fig:time:lbfgs} shows the average running time needed for L-BFGS to reach a gradient norm tolerance of $10^{-3}$.
Figure~\ref{fig:time:deg} shows the average sum of the number of unique degrees across all the bins, which is a measure of the number of variables that need to be optimized over (Section~3.3).
Theorem~1 shows that this number scales as $O(\sqrt{kn})$ and Figure~\ref{fig:time:deg} indeed provides slight evidence for this.

\section{Related Work \& Concluding Discussion}
\label{sec:related}
Although ERGs have a long history \cite{Strauss, NewmanComplexNetworks} and have been succesfully used as network models \cite{Goodreau, HunterGoodness, ERGmissingdata}, they usually suffer from limited scalability and degeneracy \cite{Handcock}.
A theoretical explanation for degeneracy is given by \cite{chatterjee2013} for dense graphs.
The authors of \cite{DERGMs} try to resolve degeneracy by limiting the support of an ERG.
Recently, \cite{FastEstimationByshkin} proposed slightly more scalable parameter estimates by exploiting properties of Markov chains at equilibrium, allowing for inference on a network with $10^5$ nodes.
Our paper circumvents both shortcomings by approximating dyadic independence models by means of block-approximated feature matrices, allowing for qualitative inference on sparse graphs with millions of nodes.
With sensible feature selection, degeneracy can be avoided (Section~\ref{sec:featureleakage}), and block-approximating typically results in additional probability smoothing.

The authors from \cite{DuijnERG} argued that the properties of the pseudo likelihood for analyzing social networks are poorly understood.
In this paper, we have argued that dyadic independence models can be valuable by themselves.
A possible explanation of the strength of the proposed models is the combination of both local role-based similarities (degrees) \cite{Rossi2019FromCT}, with community information.
Leveraging other graphlets counts, e.g., by considering weighted common graphlet counts as features, is a promising avenue for further work. Recent methods \cite{Hone, EstimationHO} have been proposed to obtain approximative low-rank decompositions of such matrices.

This paper essentially provides a general framework on how to fit maximum entropy models with local and global constraints.
The main idea is to smoothen the global constraints into regions of the graph where they are roughly the same, and then solving a smaller optimization problem by noting that the local constraints have limited uniqueness in those regions.
Including for node- or edge attributed graphs and other types of graphs is an open avenue for further work.


\section*{Acknowledgments}
The research leading to these results has received funding from the European Research Council under the European Union's Seventh Framework Programme (FP7/2007-2013) / ERC Grant Agreement no. 615517, from the Flemish Government under the ``Onderzoeksprogramma Artificiële Intelligentie (AI) Vlaanderen'' programme, and from the FWO (project no. G091017N, G0F9816N, 3G042220).

\bibliography{MAXENT_DSAA}
\bibliographystyle{ieeetr}

\clearpage

\appendix

\ptitle{Appendix A [Model derivation]}

When considering the Lagrangian of the optimization problem (\ref{eq:maxentGeneral}), and setting derivatives with respect to $P(\textbf{A})$ equal to zero, one obtains the following form for $P(\textbf{A})$:
\begin{align}
\label{formpseudo}
P(\textbf{A}) = \dfrac{\exp(\textstyle\sum_{l}\lambda_l \textstyle\sum_{i,j}f^l_{ij}\textbf{A}_{ij})}{Z} = \dfrac{\prod_{i,j}\exp(\textstyle\sum_{l}\lambda_lf^l_{ij}\textbf{A}_{ij})}{Z},
\end{align}
where $Z=Z(\lambda_1,\ldots,\lambda_M)$ is the partition function (i.e., the normalizing constant). Evaluating $Z$ yields
\begin{align*}
Z &= \textstyle\sum_{\textbf{A} \in \{0,1\}^{n\times n}} \prod_{i,j}\exp(\textstyle\sum_{l}\lambda_lf^l_{ij}\textbf{A}_{ij}) \\
&= \textstyle \prod_{i,j}\sum_{\textbf{A}_{ij} \in \{0,1\}} \exp(\textstyle\sum_{l}\lambda_lf^l_{ij}\textbf{A}_{ij}) \\
&= \textstyle \prod_{i,j}\big(1+\exp(\textstyle\sum_{l}\lambda_lf^l_{ij})\big).  
\end{align*}
Hence $P(\textbf{A})$ factorizes as a product of independent Bernoulli distributions.

On the other hand, the pseudolikelihood (see e.g., \cite{DuijnERG}) of an ERG (\ref{eq:ergintro}) with parameters $\theta$ and related statistics $s$ is equal to 
\begin{align*}
P_{\text{pseudo}}(\textbf{A}) =  \prod_{i,j}\dfrac{\exp(\textstyle\sum_{l}\theta_l\Delta s^l_{ij}\textbf{A}_{ij})}{1+\exp(\textstyle\sum_{l}\theta_l\Delta s^l_{ij})},
\end{align*} 
where $\Delta s^l_{ij}$ denotes the change in the statistic $s^l$ when going from a realization of $\textbf{A}$ without the edge $(i,j)$ to a realization of $\textbf{A}$ that includes the edge $(i,j)$.
Hence $P_{\text{pseudo}}$ is identical to (\ref{formpseudo}), when the $f^l_{ij}$ features are chosen to be equal to $\Delta s^l_{ij}$, as will be the case for statistics that are \emph{counts} of certain graph related properties (degrees, triangles, etc.).
For example, consider the out-degree of node $i$ as a statistic. Then $\Delta s^l_{kj} = 1$ when $k=i$ and $j \neq i$, and $\Delta s^l_{kj} = 0$ elsewhere.
This is exactly equal to the $f^l_{ij}$ degree-features defined in Section~\ref{sec:genmodel}.

\ptitle{Appendix B [Experimental setup link prediction]}
The sets of train and test edges are topped with equal amounts of non-edges drawn uniformly at random from the original graph. The train set is further divided in 90\% train and 10\% validation for hyper-parameter tunning.

We compare with the following embedding based methods:

\squishlist
\item \emph{Deepwalk} \cite{perozzi2014deepwalk} computes node embeddings using truncated random walks and the Skipgram model approximated by hierarchical softmax.

\vfill\break

\vfill\break

\item \emph{Node2vec} \cite{grover2016node2vec} is a generalization of DeepWalk which approximates the Skipgram model by means of negative sampling. The model uses two parameters $p$ and $q$ that interpolate between the importance of BFS/DFS random walk strategies.

\item \emph{Struc2vec} \cite{ribeiro2017struc2vec} extracts structural information from the graph through node pair similarities for a range of neighbourhood sizes. This information is then summarized as a multi-layer weighted graph. A random walk on this graph is used to generate the embeddings. 

\item \emph{Role2vec} \cite{Ahmed18Role2vec} is a space-efficient random walk based approach which learns embeddings for different types of nodes.

\item \emph{LINE} \cite{tang2015line} is a probabilistic approach which computes node embeddings based on first and second order similarities between network nodes.

\item \emph{SDNE} \cite{wang2016sdne} uses a deep auto-encoder to generate embeddings which capture first and second order proximities.

\item \emph{CNE} \cite{kang2018cne} uses a Bayesian approach to generate embeddings using structural graph properties as priors. 

\item \emph{AROPE} \cite{Zhang2018ArbitraryOrderPP} proposes embeddings as found by the truncated singular value decompositions of higher order proximities.

\squishend
The link prediction heuristics and embedding methods AROPE, CNE and MaxEnt provide node similarities which can be directly interpreted as probabilities of linking nodes. For the remaining methods we apply Logistic Regression with 5-fold cross validation on the edge embeddings to obtain link predictions.

We set $d=128$ for all methods. The following method parameters were tuned using grid search: for DeepWalk and Struc2vec window sizes $\{5,10,20\}$, for Node2vec and Role2vec $p=q \in \{0.5,1,2\}$ and window sizes $\{5,10,20\}$, for LINE learning rate $\rho \in \{0.01, 0.025\}$ and number of negative samples $\{5M, 10M\}$, for SDNE $\beta \in \{2,5,10\}$, for AROPE higher-order proximities of orders 1 up to 4 and for CNE the learning rate $\alpha \in \{0.01, 0.05\}$.

Most  methods require a binary operator that transforms node embeddings into edge embeddings. We use the same operators proposed in \cite{grover2016node2vec}, and tuned them as additional method hyper-parameters.

\end{document}